\documentclass[conference]{IEEEtran}
\IEEEoverridecommandlockouts
\usepackage{amsmath, amsthm, amssymb}
\usepackage{fp,float}
\usepackage{dsfont}
\newtheorem{lemma}{Lemma}
\newtheorem*{lemma*}{Lemma}
\newtheorem{assumption}{Assumption}
\usepackage{comment}
\usepackage{algorithm}
\newcommand\norm[1]{\left\lVert#1\right\rVert}
\newtheorem{theorem}{Theorem}
\newtheorem{remark}{Remark}

\newtheorem{definition}{Definition}

\usepackage{xcolor}

\usepackage{algpseudocode}%
\def\BibTeX{{\rm B\kern-.05em{\sc i\kern-.025em b}\kern-.08em
		T\kern-.1667em\lower.7ex\hbox{E}\kern-.125emX}}
\usepackage[noadjust]{cite}
\usepackage{pgf,pgfplots,tikz,tikz-inet}
 
\usepackage{ucs}
\usepackage[utf8x]{inputenc}
\usepackage{pgfplots}
\pgfplotsset{compat=1.7}
\usepackage{caption}
\usepackage{url}     
\usepackage{stfloats}   
\usepackage{stmaryrd}
\usepackage{lipsum}
\usepackage{algorithmicx}
\usepackage{subfigure}

\usepackage{array}
\usepackage{babel}
\usepackage{algorithm}
\usepackage{algorithmicx}
\usepackage{algpseudocode}  
\usepackage{lmodern}
\usepackage{cleveref}
\usepgfplotslibrary{colorbrewer}

\newcommand{\myquotation}[1]{\begin{center}\textit{``#1''}\end{center}}

\newcommand{\myitem}[1]{\vspace{0.25\baselineskip}\noindent\textbf{#1}}

\begin{document}
	\title{Distributed no-regret edge resource allocation with limited communication}
	
	 \author{\large{Saad Kriouile, Dimitrios Tsilimantos, and Theodoros Giannakas}\\
	 \normalsize
	 Huawei Paris Research Centre,~France, first.last@huawei.com
	 
	\thanks{This work was supported by the CHIST-ERA LeadingEdge project, grant CHIST-ERA-18-SDCDN-004 through ANR grant number ANR-19-CHR3-0007-06. The three authors contributed equally; the order is random.}}
		
	\maketitle
	
	\begin{abstract}
To accommodate low latency and computation-intensive services, such as the Internet-of-Things (IoT), 5G networks are expected to have cloud and edge computing capabilities. To this end, we consider a generic network setup where devices, performing analytics-related tasks, can partially process a task and offload its remainder to base stations, which can then reroute it to cloud and/or to edge servers. 
To account for the potentially unpredictable traffic demands and edge network dynamics, we formulate the resource allocation as an online convex optimization problem with service violation constraints and allow limited communication between neighboring nodes. 
To address the problem, we propose an online distributed (across the nodes) primal-dual algorithm and prove that it achieves sublinear regret and violation; in fact, the achieved bound is of the same order as the best known centralized alternative. Our results are further supported using the publicly available Milano dataset.
	\end{abstract}
	
	\begin{IEEEkeywords}
		Online convex optimization, edge computing, resource allocation, distributed algorithms.
	\end{IEEEkeywords}
	
	\section{Introduction}
	\label{sec:intro}
	\subsection{Motivation}

It is envisioned that globally more than 29.3 billion networked devices will be connected to the Internet of Things (IoT) by 2023~\cite{cisco2020cisco}, offering automation and real-time monitoring of machine and human-driven processes.
A main challenge in IoT deployment lies with the massive amount of connected devices; and in particular with the device heterogeneity (e.g., different computational capabilities) and the diverse and potentially stringent service (task) requirements~\cite{chen2017online}.
To host the unprecedented IoT data traffic, the edge computing paradigm has recently gained a lot of momentum as complementary to that of the cloud and it has been deemed as a key enabler to what ultimately will be the so-called ``Cloud-to-Things Continuum''~\cite{chiang2016fog,chouayakh2022towards}.

In that framework, a plethora of spatially distributed devices collect data from sensors and perform low-latency impromptu computation (e.g., machine-learning inference) using energy-limited resources. The envisaged ``Cloud-to-Things Continuum'' allows flexible task offloading \emph{from} an IoT device, via base stations, \emph{towards} more computationally powerful edge servers and, if needed, to the cloud. 
Although this architecture is promising, allocating resources for IoT computations has two distinct fundamental challenges: (a) the resources are allocated in the presence of highly unpredictable and non-stationary request patterns (demands) and network conditions; (b) the network nodes handling those tasks, namely devices, base stations and servers, are distributed and should act in the absence of a centralized entity with full observability. Naturally, the following question arises:

\myquotation{Can we offer an efficient distributed data-driven algorithm for resource allocation in the IoT context?}

In order to address this question, in this paper we consider a distributed setting with nodes of different capabilities and employ Online Convex Optimization (OCO)~\cite{zinkevich2003online}. The use of OCO is suitable for problems that are difficult to model due to unpredictable dynamics and provides concrete theoretical guarantees for such problems even in the worst-case.

\subsection{Related Work}

The related literature can be split into two categories. 
The first corresponds to studies of \emph{IoT network optimization} in the edge-cloud setting, using similar system model and assumptions to ours. 
An offline version of the problem is formulated and then decomposed across its different domains (fog and cloud) in~\cite{deng2015towards}, resulting in convex subproblems. In~\cite{lee2017online}, the authors consider the latency minimization and develop an algorithm based on the online secretary framework---however, no constraints are used in their formulation. 
Closer to our work, \cite{chen2017online, chouayakh2022towards} formulate the resource allocation as an OCO problem, and model the service violations as long-term constraints.
The learning rate is adapted to the different IoT tasks in~\cite{chen2017online}; and further extending the notion of constraint violation (see~\cite{yuan2018online}), the number of violations is also considered in~\cite{chouayakh2022towards}. Unlike our work, both approaches are centralized.

The second one deals with works on \emph{OCO with constraints in generic settings}, such as ~\cite{mahdavi2012trading, neely2017online, liakopoulos2019cautious}. Although generic, these works do not apply to our problem as they develop centralized algorithms. 
Delayed feedback on the cost and constraint functions arrive to a centralized agent in~\cite{cao2020constrained}; our approach adopts a different feedback model based on limited exchange of information between nodes. 
Finally, a distributed OCO algorithm in an environment with time-varying constraints is presented and analyzed in~\cite{yi2020distributed}, but, unlike our approach, the nodes/actors use synchronous information and make consensus steps.

\subsection{Contributions and Structure}

In this work, we approach the resource allocation in an edge-cloud scenario in a distributed way; our main contributions can be summarized as follows.

\textbf{(C.1)} We model the resource allocation as a distributed constrained OCO problem. The network nodes (devices, base stations, edge servers) are cast as individual OCO agents that must collectively optimize a given network performance metric and  satisfy service requirement guarantees (modeled as long-term constraints). To this end, we define a limited communication model between neighboring agents that allows them to exchange crucial information, related to their coupling constraints. 

\textbf{(C.2)} We propose an online primal-dual algorithm, based on projected gradient descent, with a sub-linear regret bound $\mathcal{O}(T^{1/2})$ and a sub-linear constraint violation bound $\mathcal{O}(T^{3/4})$. These bounds are equivalent to a centralized approach besides a multiplicative factor. We validate our theoretical results with numerical simulations on a commonly used dataset, and compare the performance of our algorithm to benchmarks.

The remainder of this paper is organized as follows. We summarize our system model and assumptions in Section \ref{sec:problem-setup}. Then, we introduce the OCO formulation of the problem and present our proposed algorithm and its theoretical guarantees in Section \ref{sec:oco-formulation}. We show our numerical results in Section \ref{sec:perf} and conclude the paper in Section \ref{sec:conclusions}.

	\section{Problem Setup}
	\label{sec:problem-setup}

In this section we present our network model and main assumptions. In particular, we start by the edge computing components that are available in our IoT application; then we present the control (optimization) variables, and finally we discuss our performance objectives and system constraints. 

In what follows, we use bold fonts for vectors and matrices, and $\mathcal{A}$ for a set with $|\mathcal{A}| = A$ elements.

\subsection{Topology and Computational Requests}\label{subsec:network-components}

We consider a layer of IoT sensors, which receive computational requests (e.g., analytics tasks) that need to be executed, similarly to~\cite{chen2017online, chouayakh2022towards, souza2016handling, deng2015towards, yousefpour2017fog, lee2017online}. Time is slotted and at every timeslot $t$ those requests arrive to a set of devices $\mathcal{D}$. We denote the vector of requests as $\mathbf{r}^t\in\mathbb{R}_+^D$.
Before $\mathbf{r}^t$ is revealed, the Network Operator (NO) has to reserve resources across its network infrastructure in order to accommodate them. 

We assume that the network consists of the following nodes, as shown in Fig.~\ref{fig:topology}.

\begin{itemize}
\item Devices at the edge, denoted by $\mathcal{D}$;
\item Base Stations (BSs) at the edge, denoted by $\mathcal{B}$;
\item Servers at the edge, denoted by $\mathcal{S}$;
\item A cloud server at the core network, denoted by $C$.
\end{itemize}
Throughout the rest of this work, we will refer to the above entities (except for the cloud) as ``nodes'' or ``agents''. We denote by $\mathcal{N}$ the set of all nodes across the network with $N=D+B+S$.

Each device can process locally part of the computation and can also offload tasks via a wireless channel to the BSs.
Then, the BSs can forward an incoming task either to the edge servers (wirelessly) or to the cloud. Finally, each edge server can process the received tasks locally, reroute to other edge servers or forward to the cloud; the latter only executes tasks. For simplicity, we assume full connectivity between nodes, but our methodology applies to any connectivity graph.

\begin{figure}[t]
	\centering
	\includegraphics[width=0.55\columnwidth,trim={1cm 1cm 1.2cm 1.3cm},clip]{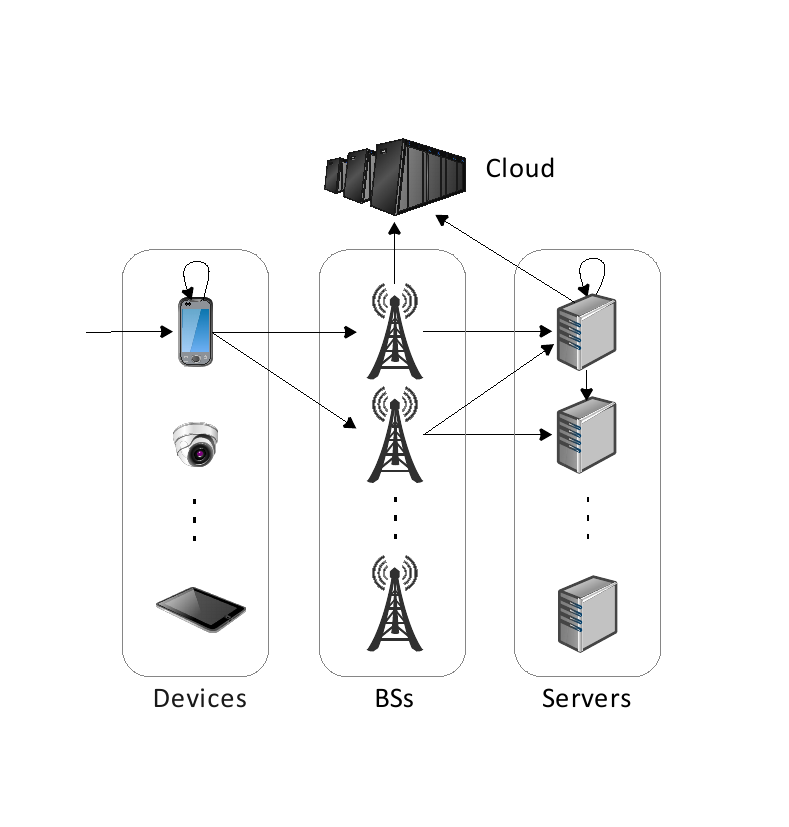}
	\caption{Topology of our edge computing setting}
	\label{fig:topology}
\end{figure}

\subsection{Distributed Control Variables}\label{subsec:control-variables}

The NO wishes to optimize a set of performance metrics in a \emph{distributed} manner.
Therefore, we wish to design a system where each agent decides its own actions. At every $t$, the control variables for every device $d\in\mathcal{D}$, BS $b\in\mathcal{B}$ and server $s\in\mathcal{S}$ are
\begin{align}\label{eq:device-var}
\mathbf{x}_d^t &= [w^t_{d0}, w^t_{d1}, \dots, w^t_{dB}, p^t_{d1}, \dots, p^t_{dB}] \in\mathbb{R}_+^{2B+1} \\
\mathbf{x}_b^t &= [y^t_{bC}, y^t_{b1}, \dots, y^t_{bS}, q^t_{b1}, \dots, q^t_{bS}]\in\mathbb{R}_+^{2S+1} \\
\mathbf{x}_s^t &= [z^t_{sC}, z^t_{s1}, \dots, z^t_{sS}] \in\mathbb{R}_+^{S+1}.
\end{align}

For each device $d$, $w^t_{d0}$ is the locally executed tasks and $w^t_{db}$, $p^t_{db}$ are the offloaded tasks and the respective transmission power to each available BS. For each BS $b$, $y_{b,C}^t$ denotes the tasks offloaded to the cloud and $y_{bs}^t, q_{bs}^t$ are the offloaded tasks and the respective transmission power to each available server $s\in\mathcal{S}$. Finally, for each server $s$, $z^t_{sC}$ is the offloaded tasks to the cloud and $z^t_{ss'}$ is the offloaded tasks to each available server $s'\in\mathcal{S}$, including $z^t_{ss}$ that denotes the locally processed tasks.

For every $t$, it must hold $\mathbf{x}_d^t\in\Omega_d, \mathbf{x}_b^t\in\Omega_b, \mathbf{x}_s^t\in\Omega_s$, where 
$\Omega_d, \Omega_b, \Omega_s$ are time-invariant box constraints of the form $\{\mathbf{x}: \mathbf{0}\le \mathbf{x} \le \mathbf{\bar{x}}\}$. Note that these constraints are local, meaning that an agent needs no external information in order to satisfy them. 
To denote the \emph{collective} control variable of all nodes, we use the following notation:
\begin{equation}
\mathbf{x}^t_{\mathcal{D}} = \{\mathbf{x}^t_d\}_{\forall d\in\mathcal{D}},~~
\mathbf{x}^t_{\mathcal{B}} = \{\mathbf{x}^t_b\}_{\forall b\in\mathcal{B}},~~
\mathbf{x}^t_{\mathcal{S}} = \{\mathbf{x}^t_s\}_{\forall s\in\mathcal{S}}.
\end{equation}
Finally, we use $\mathbf{x}^t = \{\mathbf{x}^t_{\mathcal{D}}, \mathbf{x}^t_{\mathcal{B}}, \mathbf{x}^t_{\mathcal{S}}\} \in \Omega \subset \mathbb{R}_+^V$ to denote all the variables across the network, with $V=D(2B+1)+B(2S+1)+S(S+1)$.

\subsection{Performance Objectives}

Our main target is to choose the resources $\mathbf{x}^t$, such that a cumulative cost for the total delay and transmit power across all nodes is minimized. 
The cost function at each node depends on its control variables and it is considered time-varying  in order to capture the unpredictable network dynamics at every timeslot, e.g. the randomness of the wireless channels or the network congestion levels. More precisely,  we have the following cost functions per node:
\begin{align}
	f_{d}^t(\mathbf{x}_{d}^{t})&=c^t_{d}(w_{d0}^t) + \sum_{b \in B} c^t_{db}(w_{db}^t,p_{db}^t)  \\
	f_{b}^t(\mathbf{x}_{b}^{t})&=  \sum_{s \in S} c^t_{bs}(y_{bs}^t,q_{bs}^t)+c^t_{bC}(y_{bC}^t) \\
	f^t_{s}(\mathbf{x}_{s}^{t})&= \sum_{s' \in S} c^t_{ss'}(z_{ss'}^t) + c^t_{sC}(z_{sC}^t).  
\end{align}

The functions $c^t_{d}(.), c^t_{ss}(.)$ represent local processing delay cost, $c^t_{db}(.), c^t_{bs}(.)$ capture both delay and power cost for the wireless links between nodes, $c^t_{bC}(.), c^t_{sC}(.)$ are used for the offloading delay cost to the cloud and $c^t_{ss'}(.)$ introduces the delay cost for the wired links between servers. 
At every timeslot $t$, the total cost is expressed as
\begin{equation}
f^t(\mathbf{x}^t) = \sum_{d\in\mathcal{D}} f^t(\mathbf{x}^t_d) + \sum_{b\in\mathcal{B}} f^t(\mathbf{x}^t_b) + \sum_{s\in\mathcal{S}} f^t(\mathbf{x}^t_s).
\end{equation}

Similar to our model, most related works, e.g.,~\cite{chen2017online, chouayakh2022towards, lee2017online}, assume that local processing delay and communication delay, often specified via standard queuing models, are expressed by functions of only local control variables. As we will see next, this is not the case for the problem constraints, where there exist constraints that couple agents to preserve the flow of tasks inside the network.

\subsection{Constraints}

We now focus on the constraints imposed by our application. In practice, a good decision must first ensure that the incoming tasks are either processed locally or offloaded to other nodes (\emph{flow conservation constraint}). Moreover, the transmission data rate of a wireless link must be sufficient for the assigned offloaded tasks (\emph{rate constraint}). We model these rates using the well known Shannon capacity. All constraint functions are considered time-varying in order to model the unknown dynamics of incoming tasks and channel gains. Given that the agents first reserve resources and then the tasks are revealed, it is possible to have service violations, i.e. the provisioned resources are not adequate or cannot be realized.

For each device $d \in \mathcal{D}$, the constraint functions are
\begin{align} 
	g^t_{d0}(\mathbf{x}_{d}^{t})& = r^t_{d} - w^t_{d0} - \sum_{b \in \mathcal{B}} w^t_{db}\label{eq:flow-con-dev} \\
	g_{db}^t(\mathbf{x}_{d}^{t})& = w^t_{db} - b_w\log_2\big(1 + \alpha_{db}^t p_{db}^t), \forall~b\in \mathcal{B},
\end{align}
where $b_w$ is a constant for the transmission bandwidth and $\alpha_{db}^t$ is an unknown variable for the channel gain, including the effect of interference and noise. 
In a similar way, for each BS $b \in \mathcal{B}$ we have 
\begin{align}
g^t_{b0}(\mathbf{x}_{b}^{t}; \mathbf{x}_{\mathcal{D}}^{t})& =\sum_{d \in \mathcal{D}} w^t_{db} - y^t_{bC} - \sum_{s\in\mathcal{S}} y^t_{bs}\label{eq:flow-con-bs}\\
g_{bs}^t(\mathbf{x}_{b}^{t})& = y^t_{bs} - b_w\log_2\big(1 + \alpha_{bs}^t q_{bs}^t), \forall~s\in \mathcal{S},
\end{align}
where $\alpha_{bs}^t$ is defined as $\alpha_{db}^t$ above. 
Notice that \eqref{eq:flow-con-bs} is a \emph{coupling constraint}, i.e. to evaluate the function at BS $b$, we need to know the external variables $\{w^t_{db}\}_{d\in\mathcal{D}}$ of devices. Conventionally, we denote this dependency with conditional arguments to distinguish from locally available variables, as denoted by $g^t_{b0}(\mathbf{x}_{b}^{t}; \mathbf{x}_{\mathcal{D}}^{t})$.
Finally, for each server $s\in \mathcal{S}$, we have only the following flow conservation constraint function
\begin{align}
g_{s0}^t(\mathbf{x}_{s}^{t}; \mathbf{x}_{\mathcal{B}}^{t}, \mathbf{x}_{\mathcal{S}_{-s}}^{t}) \! = \! \sum_{b\in\mathcal{B}} y^t_{bs} + \!\!\! \sum_{s' \in \mathcal{S}_{-s}}\!\!\! z^t_{s's} - \!\! \sum_{s'\in\mathcal{S}}\! z^t_{ss'} \! - \! z^t_{sC} \label{eq:flow-con-server}
\end{align}
where $\mathcal{S}_{-s}$ is used to denote the set of edge servers $\mathcal{S}$ excluding $s$.
\eqref{eq:flow-con-server} is also a coupling constraint, since server $s$ needs to know the external variables $\{y^t_{bs}\}_{b\in\mathcal{B}}$ of BSs and $\{z^t_{s's}\}_{s'\in\mathcal{S}_{-s}}$ of other servers.

To denote the collective set of constraints per device $d$ and per BS $b$, we use the notation
\begin{align}
	\mathbf{g}_d^t(\mathbf{x}_{d}^{t}) &=  g^t_{d0}(\mathbf{x}_{d}^{t}) \cup  \{g^t_{db}(\mathbf{x}_{d}^{t})\}_{\forall b\in\mathcal{B}} \\
	\mathbf{g}_b^t(\mathbf{x}_{b}^{t};\mathbf{x}_{\mathcal{D}}^{t}) &=  g^t_{b0}(\mathbf{x}_{b}^{t};\mathbf{x}_{\mathcal{D}}^{t}) \cup  \{g^t_{bs}(\mathbf{x}_{b}^{t})\}_{\forall s\in\mathcal{S}} \label{eq:constraint_g_last}
\end{align}
and write $\mathbf{g}^t(\mathbf{x}^t) = \lbrace \{ \mathbf{g}^t_{d}\}_{d\in \mathcal{D}}, \{\mathbf{g}^t_{b}\}_{d\in \mathcal{B}}, \{g^t_{s0}\}_{s\in \mathcal{S}} \rbrace: \mathbb{R}_+^{V} \to \mathbb{R}^{M}$ to denote all constraints across the network, where $M=B(D+S)+N$.

\subsection{Exchange of Information}

Due to the coupling constraints, a fully distributed solution is no longer possible. To circumvent this, we allow each node to exchange information with other nodes, so that it can take into account the coupling constraints in its local decisions. 
The exact messages to exchange is part of the algorithm design, which ideally should achieve a performance close to a centralized approach with a minimal exchange of information. To limit the communication overhead, we further consider that each node can send information to other nodes only once during a timeslot, i.e. at the beginning of a timeslot. As we will see in the next section, this practical assumption introduces delayed feedback between nodes, as not all of the required information is available for a single transmission step within a timeslot.

	\section{OCO Formulation and Decentralized Algorithm}
	\label{sec:oco-formulation}
	The goal of this section is to formulate the resource allocation, described in the previous section, as an Online Convex Optimization (OCO) problem~\cite{zinkevich2003online} and propose an algorithm to solve it. Importantly, we explain how we adapt a centralized algorithm to transcend towards one which can run in a distributed fashion; finally, we present the performance guarantees of our solution.

\subsection{OCO Preliminaries}\label{subsec:oco-form}

We start with the basics of OCO in a centralized algorithm for two reasons: (a) it helps to introduce useful concepts and metrics; and (b) we benchmark our distributed algorithm with respect to a centralized one.
To formulate the resource allocation as an OCO problem, we first need to define the sequence of events taking place for the central agent during every timeslot $t$. 
\begin{enumerate}
	\item The agent implements an action $\mathbf{x}^t$. 
	\item The environment reveals all the unknown variables, e.g. computation requests $\mathbf{r}^t$ and channel gains $\alpha_{ij} ^t$, which in the OCO framework can be random variables or even controlled by an adversary. Using these, the functions $f^t(\mathbf{.})$ and $\mathbf{g}^t(\mathbf{.})$ become known. 
	\item The agent receives or evaluates cost and constraint violations, i.e. the values  $f^t(\mathbf{x}^t)$ and $\mathbf{g}^t(\mathbf{x}^t)$. 
	\item The agent updates its action to $\mathbf{x}^{t+1}$.
\end{enumerate}

Below, we define the benchmark actions and the metrics to evaluate an algorithm that produces a sequence of actions $\{\mathbf{x}^t\}_{t=1,\dots,T}$.
\begin{definition}[Static Regret]
The fixed optimal action $\mathbf{x}_{*}$ and the static regret are defined as
	\begin{gather}
	\mathbf{x}_{*} = \underset{x\in\Omega}{\mathrm{argmin}} \sum_{t=1}^T f^t(\mathbf{x}), \emph{s.t.}~\mathbf{g}^t(\mathbf{x}) \le 0,~t=1,\dots,T \label{eq:static-bench}\\
	\emph{Reg}_{S}(T) = \sum_{t=1}^T f^t(\mathbf{x}^t) - \sum_{t=1}^T f^t(\mathbf{x}_*).
\end{gather}
\end{definition}
\begin{definition}[Dynamic Regret]
The per slot optimal action  $\{\mathbf{x}_{*}^t\}_{t=1,\dots,T}$ and the dynamic regret are defined as 
\begin{gather}
	\mathbf{x}_{*}^t = \underset{x\in\Omega}{\mathrm{argmin}} f^t(\mathbf{x}), \emph{s.t.}~\mathbf{g}^t(\mathbf{x}) \le 0 \label{eq:dynamic-bench}\\
	\emph{Reg}_{D}(T) = \sum_{t=1}^T f^t(\mathbf{x}^t) - \sum_{t=1}^T f^t(\mathbf{x}_*^t).
\end{gather}
\end{definition}
\begin{definition}[Fit]
	Using the clipped constraint function $h_m^t(\mathbf{x}^t) := [g_m^t(\mathbf{x}^t)]^+=\max\{0,g_m^t(\mathbf{x}^t)\}$, the fit 
	is defined as
	\begin{equation}
		\emph{Fit}(T) = \sum_{t=1}^T \sum_{m=1}^M h_m^t(\mathbf{x}^t).
	\end{equation}
\end{definition}

The static regret is a standard metric for evaluating OCO-based algorithms; however, there is also a growing interest recently for the dynamic regret~\cite{jadbabaie2015online}, which is in principle a much harder metric. 
For the fit we use a clipped version of the constraint, i.e. we do not allow negative and positive values of the constraints to  average out in the long run.
This is a suitable modeling approach, as it would be unrealistic to assume that overprovisioning in certain timeslots can compensate for missing resources or channel rate violations in other timeslots. 
Finally, similarly to~\cite{yuan2018online}, we define the gradient of $h_m(\mathbf{x})$ as
\begin{align}
	\nabla h_m(\mathbf{x})=\nabla[g_m(\mathbf{x})]^+=
	\begin{cases}
		\mathbf{0} & \text{if $g_m(\mathbf{x}) \le 0$} \\
		\nabla g_m(\mathbf{x}) & \text{if $g_m(\mathbf{x}) > 0$}
	\end{cases}\nonumber
\end{align}
In the remainder, we use $\mathbf{h}^t$ for our analysis with subscripts that have the same meaning as for $\mathbf{g}^t$ in \eqref{eq:flow-con-dev}-\eqref{eq:constraint_g_last}.

Overall, the objective of an algorithm is to establish that Reg$_{S}(T)$, Reg$_{D}(T)$ and Fit$(T)$ are all sublinear in the time horizon $T$ \cite{zinkevich2003online}.

\subsection{Decomposition Across the Agents}

A centralized algorithm can use the Lagrange function 
\begin{equation}
	\mathcal{L}^t (\mathbf{x}^t, \boldsymbol{\lambda}^t) = f^t(\mathbf{x}^t) + \sum_{m=1}^M h_m^t(\mathbf{x}) \lambda^t_{m}
\end{equation}
and apply the primal-dual updates~\cite{yuan2018online} as follows
\begin{align}\label{eq:primal-dual-central}
	\boldsymbol{\lambda}^{t}=\frac{\mathbf{h}^{t}({\mathbf{x}}^{t})}{\eta \sigma},~  
	\mathbf{x}^{t+1}=\mathcal{P}_{\Omega}\left({\mathbf{x}}^t -\eta {\nabla}_{\mathbf{x}}\mathcal{L}^t({\mathbf{x}}^t, \boldsymbol{\lambda}^t) \right),
\end{align} 
where $\lambda^t_{m}$ is the Lagrange multiplier for constraint $m$, $\eta$ is the gradient step size, $\sigma$ is a constant and $\mathcal{P}_{\Omega}(\mathbf{x})$ is the projection of $\mathbf{x}$ onto $\Omega$.
Ideally, we want to perform these updates in a distributed way so that they are as close as possible to the updates of the centralized algorithm. For any node $n \in \mathcal{N}$ (device, BS or server), we have  
\begin{align}\label{eq:dual-node}
	\boldsymbol{\lambda}^{t}_{n} &= \frac{\mathbf{h}_{n}^{t}(\mathbf{x}_{n}^{t};\mathbf{x}_{\mathcal{E}_n}^{t})}{\eta \sigma}  \\
	\mathbf{x}_{n}^{t+1}&=\mathcal{P}_{\Omega_n} \left(\mathbf{x}_{n}^{t} - \eta \nabla_{\mathbf{x}_{n}}\mathcal{L}^t(\mathbf{x}^{t}, \boldsymbol{\lambda}^{t}) \right), \label{eq:primal-node}
\end{align}
where $\mathcal{E}_n$ is the set of nodes with variables required for node $n$, i.e. $\mathcal{E}_d = \emptyset$, $\mathcal{E}_b = \mathcal{D}$ \eqref{eq:flow-con-bs}, $\mathcal{E}_s = \mathcal{B}\cup \mathcal{S}_{-s}$ \eqref{eq:flow-con-server}, and $\boldsymbol{\lambda}^{t}$ uses the same subscripts for constraints as $\mathbf{g}^t,\mathbf{h}^t$.

We now focus on the gradients in \eqref{eq:primal-node} and write
\begin{equation}\label{eq:gradient-local_external}
	\nabla_{\mathbf{x}_{n}} \mathcal{L}^t(\mathbf{x}^{t}, \boldsymbol{\lambda}^{t}) = \mathbf{H}^t_{nL} + \mathbf{H}^t_{nE}
\end{equation}
where $\mathbf{H}^t_{nL}$ denotes the part that depends only on local cost and constraint functions at node $n$,
\begin{equation} \label{eq:gradient-node_local}
	\mathbf{H}^t_{nL} = \nabla_{\mathbf{x}_n} f_{n}^t(\mathbf{x}_{n}^{t}) + \boldsymbol{\lambda}_{n}^{t\top} \nabla_{\mathbf{x}_{n}} \mathbf{h}^t_n(\mathbf{x}_n^t;\mathbf{x}_{\mathcal{E}_n}^{t}),
\end{equation}
and $\mathbf{H}^t_{nE}$ describes the part that depends on external constraints from other nodes. For clarity, we provide the explicit expression for each type of node:
\begin{align} \label{eq:gradient-device_external}
	\mathbf{H}^t_{dE} &= \sum_{b \in \mathcal{B}}\lambda^{t}_{b0} \nabla_{\mathbf{x}_d} h^{t}_{b0}(\mathbf{x}_{b}^{t};\mathbf{x}_{\mathcal{D}}^{t})\\
	\label{eq:gradient-bs_external}
	\mathbf{H}^t_{bE} &= \sum_{s\in \mathcal{S}}\lambda_{s0}^{t} \nabla_{\mathbf{x}_b} h^{t}_{s0}(\mathbf{x}_{s}^{t}; \mathbf{x}_{\mathcal{B}}^{t},\mathbf{x}_{\mathcal{S}_{-s}}^{t})\\
	\label{eq:gradient-server_external}
	\mathbf{H}^t_{sE} &=\sum_{s'\in\mathcal{S}_{-s}} \!\!\lambda^{t}_{s'0}\nabla_{\mathbf{x}_s} h_{s'0}^{t}(\mathbf{x}_{s'}^{t}; \mathbf{x}_{\mathcal{B}}^{t},\mathbf{x}_{\mathcal{S}_{-s'}}^{t}).
\end{align}

Notice that in \eqref{eq:gradient-device_external}-\eqref{eq:gradient-server_external}, only one dimension of the gradients can be non-zero, as each node has only a single variable involved in coupling constraints of another node, i.e. $w_{db}, y_{bs}$ and $z_{ss'}$ for devices, BSs and servers.

\begin{remark}
According to the definition of $\mathbf{h}^t(.)$ and the flow constraints, the summation terms in  \eqref{eq:gradient-device_external}-\eqref{eq:gradient-server_external} simplify to $\lambda^{t}_{n0}  \mathbf{1}_{h^{t}_{n0}(\mathbf{x}_{n}^{t};\mathbf{x}_{\mathcal{E}_n}^{t})>0}$, where $\mathbf{1}_{(.)}$ is the indicator function.
\end{remark}

\subsection{Distributed Algorithm}
As one can already notice, a distributed version of the algorithm requires exchange of information at two different steps. Specifically, a node $n$ needs to send its variables to nodes that need them in their coupling constraints and then, receive the gradient-related feedback $\mathbf{H}^t_{nE}$. We now turn our focus on the adopted communication model and examine how this message exchange can be implemented. 

\begin{figure}[!b]
	\centering{\subfigure[Ideal case]{\includegraphics[width=0.48\columnwidth]{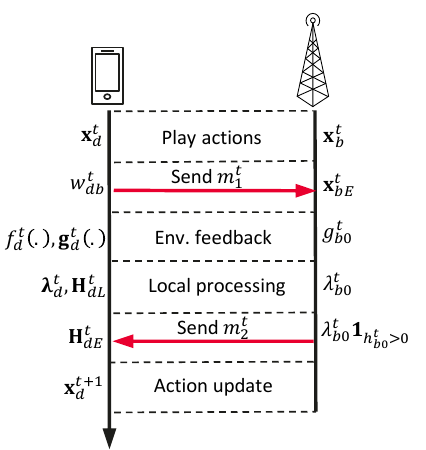}
			\label{fig:algo_ideal}}
		\hfil
		\subfigure[Proposed approach]{\includegraphics[width=0.48\columnwidth]{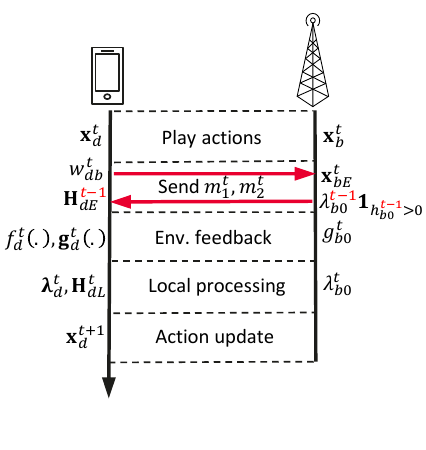}
			\label{fig:algo_real}}}
	\caption{Exchanged messages during time slot $t$ between device $d$ and BS $b$.}
	\label{fig_sim}
\end{figure}

An ideal scenario is shown in Fig.~\ref{fig:algo_ideal}, where nodes can send messages at any moment during a time slot, focusing for simplicity in the link between a device and a BS. As we can see, the BS $b$ needs to collect $\mathbf{x}^t_{\mathcal{E}_b}=\{w^t_{db}\}_{d \in \mathcal{D}}$ and perform few processing steps before it can send its feedback to the device, which can then perform the primal update $\mathbf{x}_{d}^{t+1}$. In practice, this is not possible in our model as nodes are allowed to send messages only once at the beginning of a time slot. We address this limitation by allowing nodes to send their feedback at the next time slot, as shown in Fig.~\ref{fig:algo_real}. As a result, the device uses the outdated feedback $\mathbf{H}^{t-1}_{dE}$ for its updates.

Overall, we define the following messages between nodes $n,v$, where $n \in\mathcal{E}_v$, and summarize them in Table \ref{table:messages}.

\begin{enumerate}
	\item $m^t_{1,n \to v}:=\mathbf{x}^t_{n} \in \mathbf{x}^t_{\mathcal{E}_v}$; required to evaluate the coupling constraint $h^t_{v0}(\mathbf{x}^t_v;\mathbf{x}^t_{\mathcal{E}_v})$ at node $v$, which is then used for the dual update of $\lambda^t_{v0}$ in \eqref{eq:dual-node} and the local term $\mathbf{H}^t_{vL}$.
	\item $m^t_{2,v \to n}:=\lambda^{t-1}_{v0} \mathbf{1}_{ h^{t-1}_{v0}(\mathbf{x}_{v}^{t-1};\mathbf{x}_{\mathcal{E}_{v}}^{t-1})>0}$; feedback required to evaluate $\mathbf{H}^{t-1}_{nE}$ for the primal update at node $n$. 
\end{enumerate}

\begin{table}[!t]
	\renewcommand{\arraystretch}{1.5}
	\caption{Messages  $(m^t_1,m^t_2)$ between nodes}
	\label{table:messages}
	\centering 
	\begin{tabular}{c  c  c  c  }
		\hline
		From$\backslash$To & \textbf{Device} $d'$  & \textbf{BS} $b'$  & \textbf{Server} $s'$  \\
		\hline 
		\textbf{Device} $d$  & $-$ & $w^t_{db'}$ & $-$ \\
		\textbf{BS} $b$ & $\lambda^{t-1}_{b0} \mathbf{1}_{h^{t-1}_{b0}>0}$ & $-$ & $y^t_{bs'}$ \\
		\textbf{Server} $s$ & $-$ & $\lambda^{t-1}_{s0} \mathbf{1}_{h^{t-1}_{s0}>0}$ & $z^t_{ss'}$, $ \lambda^{t-1}_{s0} \mathbf{1}_{h^{t-1}_{s0}>0} $ \vspace{1pt} \\ 
		\hline
	\end{tabular}
\end{table}

Let us now consider the primal/dual updates of the proposed approach. First, the dual update \eqref{eq:dual-node} remains the same and thus, identical to the centralized case. Then, for the primal update \eqref{eq:primal-node}, only the gradient term is modified and approximated by  
\begin{equation}\label{eq:gradient-local_external_approximation}
	\nabla_{\mathbf{x}_n} \hat{\mathcal{L}}(\mathbf{x}^t, \boldsymbol{\lambda}^t) = \mathbf{H}^t_{nL} + \mathbf{H}^{t-1}_{nE}.
\end{equation}

The steps of our proposed distributed algorithm are presented in Algorithm~\ref{algorithm} for any node $n$.  

\begin{algorithm}[t]
	\caption{Distributed OCO for node $n$}
	\label{algorithm}
	\begin{algorithmic}[1]
		\State Initialize: parameters $\sigma$, $\eta$
		\State Set first action $\mathbf{x}^{1}_{n}\in\Omega_n$ and define $\boldsymbol{\lambda}^{0}_{n}=\mathbf{0}$.
		\For{$t = 1, \dots,T$} 
		\State Play action $\boldsymbol{x}^{t}_{n}$
		\State Send messages $m^t_{1, n\to v}$ to nodes $v:n \in\mathcal{E}_v$
		\State Receive from environment functions $f^{t}_{n}(.)$ and $\boldsymbol{g}^{t}_{n}(.)$
		\State Receive feedback messages $m^t_{2, v\to n}$ from nodes $v$
		\State Compute $\boldsymbol{\lambda}_{n}^{t}$ with \eqref{eq:dual-node} \Comment Dual update 
		\State Update $\boldsymbol{x}_{n}^{t+1}$ with \eqref{eq:primal-node},\eqref{eq:gradient-local_external_approximation} \Comment Primal update 
		\EndFor
	\end{algorithmic}
\end{algorithm}

\subsection{Performance Guarantees}

We make the following standard assumptions, widely used in online learning literature (e.g., see~\cite{yuan2018online, liakopoulos2019cautious}), and then formally state our main theorem.
\begin{assumption}\label{ass:oco}\leavevmode

\begin{itemize}
\item (i) Set $\Omega_n$ is bounded and convex; specifically it holds that $\norm{\mathbf{x}_n - \mathbf{y}_n} \le R$, $\forall \mathbf{x}_n, \mathbf{y}_n \in\Omega_n$, for $n\in\mathcal{D,B,S}$.
\item (ii) For $t=1,\dots, T$, functions $f^t_n$ and $g^t_{n,i}$ are convex and Lipschitz with $\norm{\nabla_{\mathbf{x}_n} f^t_n}\le F^{\prime}$ and $\norm{\nabla_{\mathbf{x}_n} g_{n,i}^t}\le G^{\prime}$, for $n\in\mathcal{D,B,S}$ and $\mathbf{g}_n^t= \{g^t_{n,i}\}_{i=1,\dots,M_n}$ (with $M_n$ the number of constraints at node $n$).
\end{itemize}
\end{assumption}

Below we write a list of implications that we use for the proof of our theorem.
First, $f^t_n$ and $g^t_{n, i}$ are both bounded, i.e., $|f^t_n|\le F$, $|g^t_{n, i}| \le G^{\prime\prime}$. 
Second, since $\norm{\nabla g_{n,i}^t}\le G^{\prime}$, then $\norm{\nabla h_{n,i}^t}\le G^{\prime}$ (comes from the definition of gradient of $h$). Third, since $g^t_{n, i}$ is bounded, $h^t_{n, i}$ is also bounded by definition; hence $|h^t_{n,i}| \le G^{\prime\prime}$ and $\norm{\mathbf{h}^t}\le G$.
For simplicity we write that $|f^t_n|, \norm{\nabla f^t_n}\le F$ and $|h_{n,i}^t|, \norm{\nabla h_{n,i}^t}, \norm{\mathbf{h}_{n}^t} \le G$.
Fourth, since $g_{n,i}^t$ is convex, then $h_{n,i}^t$ is as well. 

A proof of the first and fourth implications can be found in Appendix \ref{app:implications}.

\begin{theorem}\label{theorem}
Given Assumption~\ref{ass:oco}, and $\sigma > 3KG^2$, Algorithm~\ref{algorithm} guarantees that
\begin{align}
\normalfont{\text{Reg}_S}(T) &\le \frac{R^2}{2 \eta} + \frac{2R E G^2}{\eta \sigma} + \frac{7}{2}\eta N F^2 T \triangleq U_{sr}, \\
\normalfont{\text{Reg}_D}(T) &\le U_{sr} + \frac{R}{\eta} V(\boldsymbol{x}_*^{1:T}) , \\
\normalfont{\text{Fit}}(T) &\le \sqrt{\frac{\eta \sigma }{\beta}MT(U_{sr} + 2 N F T )},
\end{align}
where $E$ is the number of edges in the network topology, $K=2D(3B+1) + 2B(3S+1) + 2S(2S-1)$, $\beta=1-\frac{3K G^2}{\sigma}$ and $V(\mathbf{x}_*^{1:T}) = \sum_{t=1}^{T}\norm{\mathbf{x}_*^t - \mathbf{x}_*^{t-1}}$. 
\end{theorem}
\begin{proof}
The proof can be found in Appendix \ref{app:proof_regret_fit_normal_constraint_exchanging_information}.
\end{proof}

An immediate implication of Theorem~\ref{theorem} is that, for step size $\eta=\mathcal{O}(T^{-1/2})$, we have
\begin{itemize}
\item $\textnormal{Reg}_S(T)=\mathcal{O}(T^{1/2})$
\item $\textnormal{Reg}_D(T)=\mathcal{O}(\max\{T^{1/2},T^{1/2}V(\mathbf{x}_*^{1:T})\})$
\item $\textnormal{Fit}(T)=\mathcal{O}(T^{3/4})$
\end{itemize}
We conclude that our distributed algorithm, although using outdated Lagrange multipliers, achieves sublinear static regret and fit; if additionally, $V(\mathbf{x}_*^{1:T})=o(T^{1/2})$, then dynamic regret is also sublinear.
In fact, we achieve the same order of bounds as the centralized algorithm in~\cite{chouayakh2022towards} (in the case where the authors ignore the outages), which tackles the same setting.
Note that choosing $\eta=\mathcal{O}(T^{-1/2})$ yields the minimum regret while preserving a sublinear fit.
	
	\section{Performance Evaluation}
	\label{sec:perf}
	\subsection{Simulation Setup}

\myitem{Topology and box constraints.} 
We assume a fully connected setting with nodes $D,B,S=2$. Moreover, the upper bounds of the control variables are as follows: for every device $d$, $\overline{w}_{d0}=2$, $\overline{w}_{db}=25$, and $\overline{p}_{db}=25$, for every BS $b$, $\overline{y}_{bC}=30$, $\overline{y}_{bs}=25$, and $\overline{q}_{bs}=25$ and for every server $s$, $\overline{z}_{sC}=50$, $\overline{z}_{ss}=15$ and $\overline{z}_{ss'}=10$.

\myitem{Costs.} 
We model the cost functions using expressions for delay  (from $M/M/1$\footnote{To avoid numerical instabilities, e.g., $M/M/1$ delay becoming infinite, we use standard convex extensions~\cite{chouayakh2022towards}.}) and power~\cite{lee2017online}. The local processing delay of node $n$ for $x$ tasks is
$c_{n}(x) = 1/(\overline{x} - x)$, where $\overline{x}$ denotes the capacity of node $n$; this cost is used to model $c^t_d(w^t_{d0})$ and $c^t_{ss}(z^t_{ss})$.
Then, the cost related to wireless offloading, i.e., $c^t_{db}(w^t_{db},p^t_{db})$ and $c^t_{bs}(y^t_{bs},q^t_{bs})$, is modeled as
$c^t_{nn'}(x, y) = 1/(R^t_{nn'}(y) - x) + \frac{1}{2}y^2$, 
where $R^t_{nn'}(y)=b_w\log_2\big(1 + \alpha_{nn'}^t y)$ is the channel rate. Finally, the delays for offloading to the cloud, i.e. $c^t_{bC}(y^t_{bC})$ and $c^t_{sC}(z^t_{sC})$, are modeled as $c_{nC}^t(x)=d_{nC}^t x$, where $d_{nC}^t$ is a time-varying unknown environment parameter.

\myitem{Unknown variables.} 
For each time slot $t$ we need: (a) channel gains $\alpha_{db}^t, \alpha_{bs}^t$, (b) cloud delay costs $d_{bC}^t, d_{sC}^t$, and (c) traffic requests $\mathbf{r}^{t}$. We model (a) and (b) as random variables sampled from $\mathcal{U}(8,15)$ and $\mathcal{U}(3,10)$ respectively. For (c), we mainly use the publicly available Milano dataset~\cite{DVN}; in particular, we extract the aggregate Internet traffic arrivals measured in MBs to $D=2$ BSs (devices in our model). We also provide results using synthetic demands, drawn from $\mathcal{U}(1,10)$.

\myitem{Metrics and Baselines.}
The performance of an online algorithm is evaluated using the static and dynamic regrets (the respective benchmarks are found using CVXPY~\cite{diamond2016cvxpy}), and the fit.
We plot these metrics for proposed Algorithm~\ref{algorithm}, which we refer to as \textit{Cooperative} algorithm, and for two more baselines. First, the \emph{Selfish} is a distributed algorithm without information exchange between nodes; i.e. $\mathbf{H}_{nE}^t = \mathbf{0}$ in~(\ref{eq:gradient-local_external}).
Second, the \emph{Centralized} algorithm assumes a controller with access to all necessary information in order to do the updates optimally. This is essentially the algorithm described in~\cite{yuan2018online}, adapted to our setting with time-varying constraints $\mathbf{h}^{t}$.

\subsection{Simulation Results}

In all our plots, the $x$-axis represents the horizon length $T$, which we vary from $0$ to $300$ time slots. For each algorithm we plot the average value across $4$ independent runs and with shade the corresponding standard deviation. Notice that all metrics are normalized by the horizon $T$.

Our setup is challenging for a distributed algorithm, as the flow conservation constraints \emph{couple} different nodes. To this end, we first investigate the Fit$(T)$ for the Milano and the synthetic datasets in Figs.~\ref{fig:fit-milano}, \ref{fig:fit-synthetic}.
The fit of the \emph{Centralized} algorithm quickly converges to zero, suggesting that it learns to play actions that respect most of the time-varying constraints; the reason is that it performs the best possible primal and dual updates with the freshest information. 
The fit of the proposed \emph{Cooperative} algorithm converges almost together with the \emph{Centralized} for the Milano demands (Fig.~\ref{fig:fit-milano}) and slightly slower for the synthetic ones (Fig.~\ref{fig:fit-synthetic}). Therefore, the modified gradients proposed in our algorithm suffice in order to satisfy the constraints in the long run.
The \emph{Selfish} baseline exemplifies the necessity of at least some information exchange between the nodes; we can see in both figures the fit increasing. This behavior is expected, as the Fit$(T)$ of \emph{Centralized} and \emph{Collaborative} is sublinear, whereas the one of \emph{Selfish} can be shown to be linear as its updates totally ignore the coupling (flow conservation) constraints.

Having discussed the ``feasibility'' aspect of the algorithms (i.e., how they perform in terms of constraints) we now focus on the objective function and in particular on the regrets in Figs.~\ref{fig:static-milano}, \ref{fig:dynamic-milano}. We plot these metrics only for the Milano dataset; but the respective plots for the synthetic one are similar.
A first observation is that the regrets of the \emph{Selfish} algorithm are the lowest among all three methods. This should not come as a surprise, since by construction, the algorithm solves a more relaxed version of the problem (ignores flow constraints) and therefore can achieve better cost values. The \emph{Centralized} algorithm has slightly higher regrets, which is justified by its effort to also satisfy the fit. Finally, our \emph{Collaborative} algorithm has regrets that also go to zero and are very close to the ones of the \emph{Centralized} solution.

\begin{figure}
 \centering
 \subfigure[Fit - Milano]{\includegraphics[width=0.48\columnwidth]{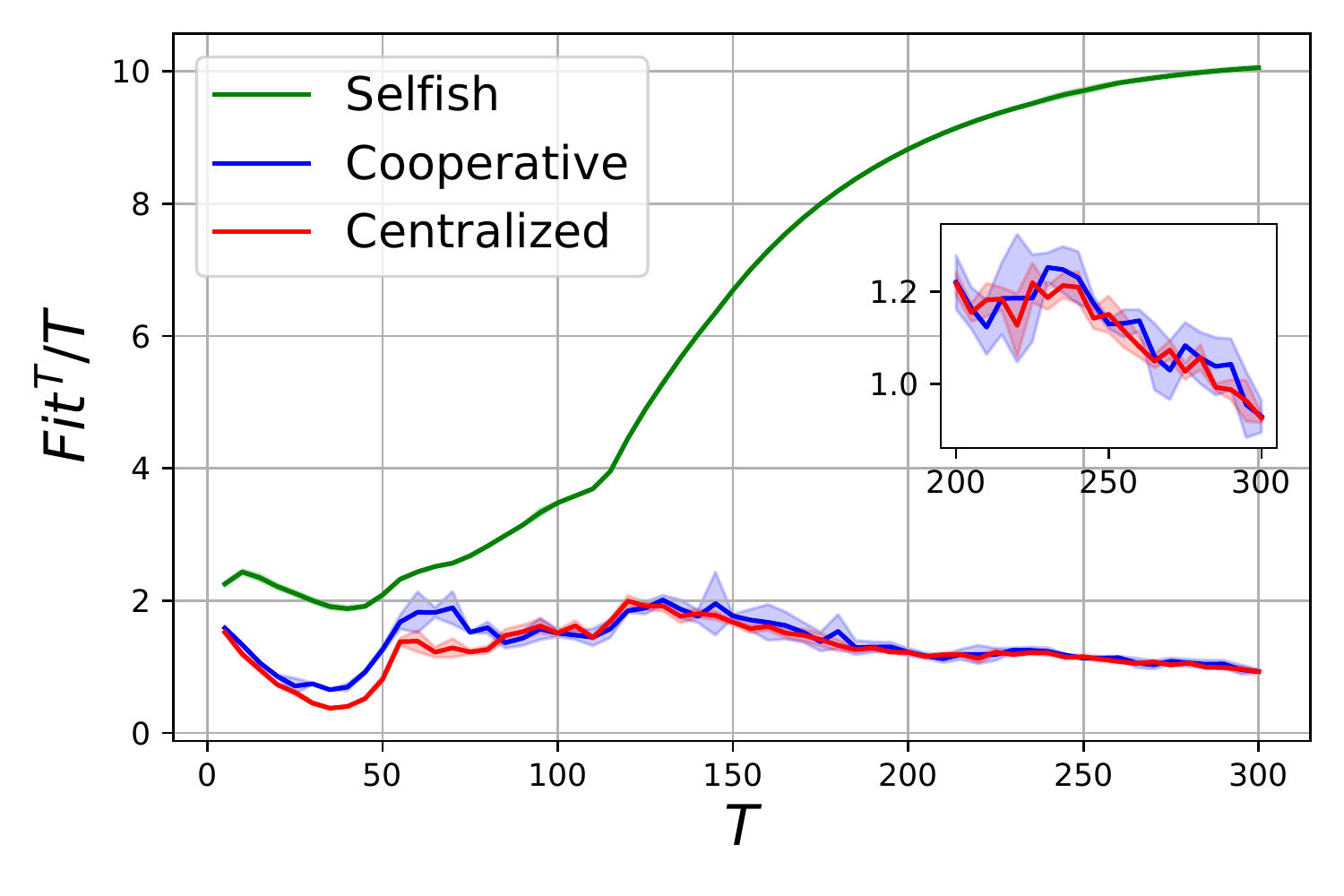}\label{fig:fit-milano}}
 \subfigure[Fit - Synthetic]{\includegraphics[width=0.48\columnwidth]{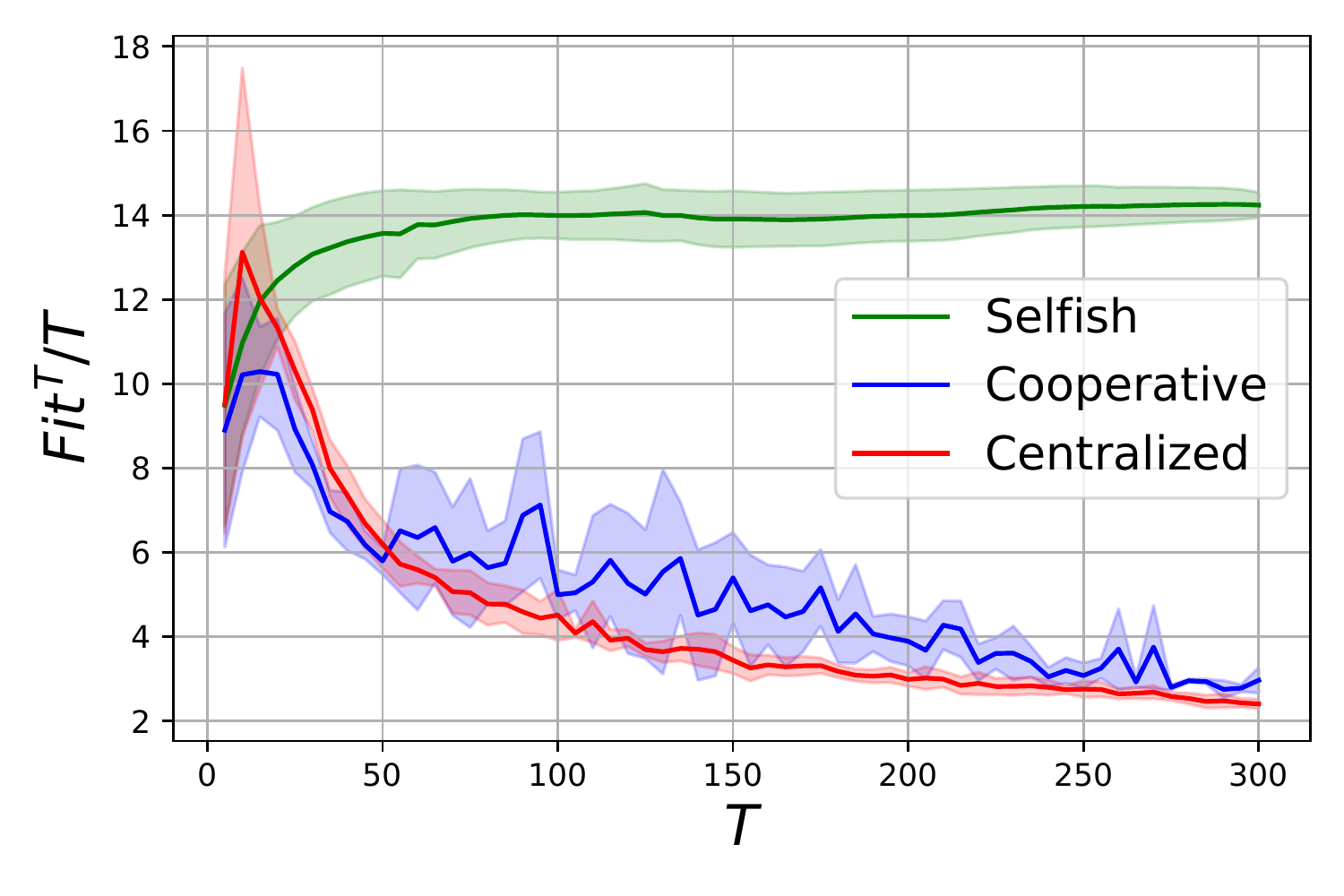}\label{fig:fit-synthetic}}
 \subfigure[Static regret - Milano]{\includegraphics[width=0.48\columnwidth]{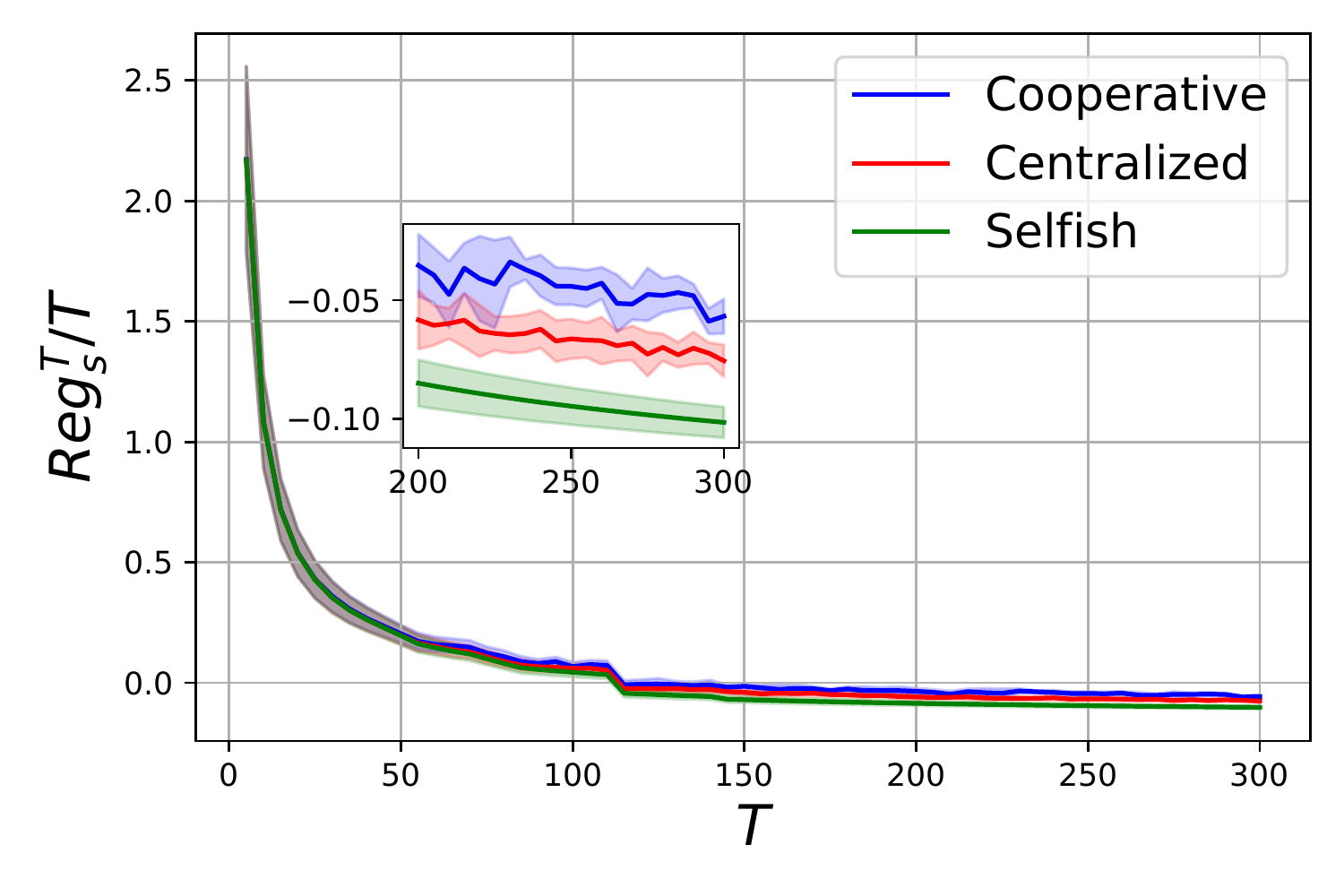}\label{fig:static-milano}}
 \subfigure[Dynamic Regret - Milano]{\includegraphics[width=0.48\columnwidth]{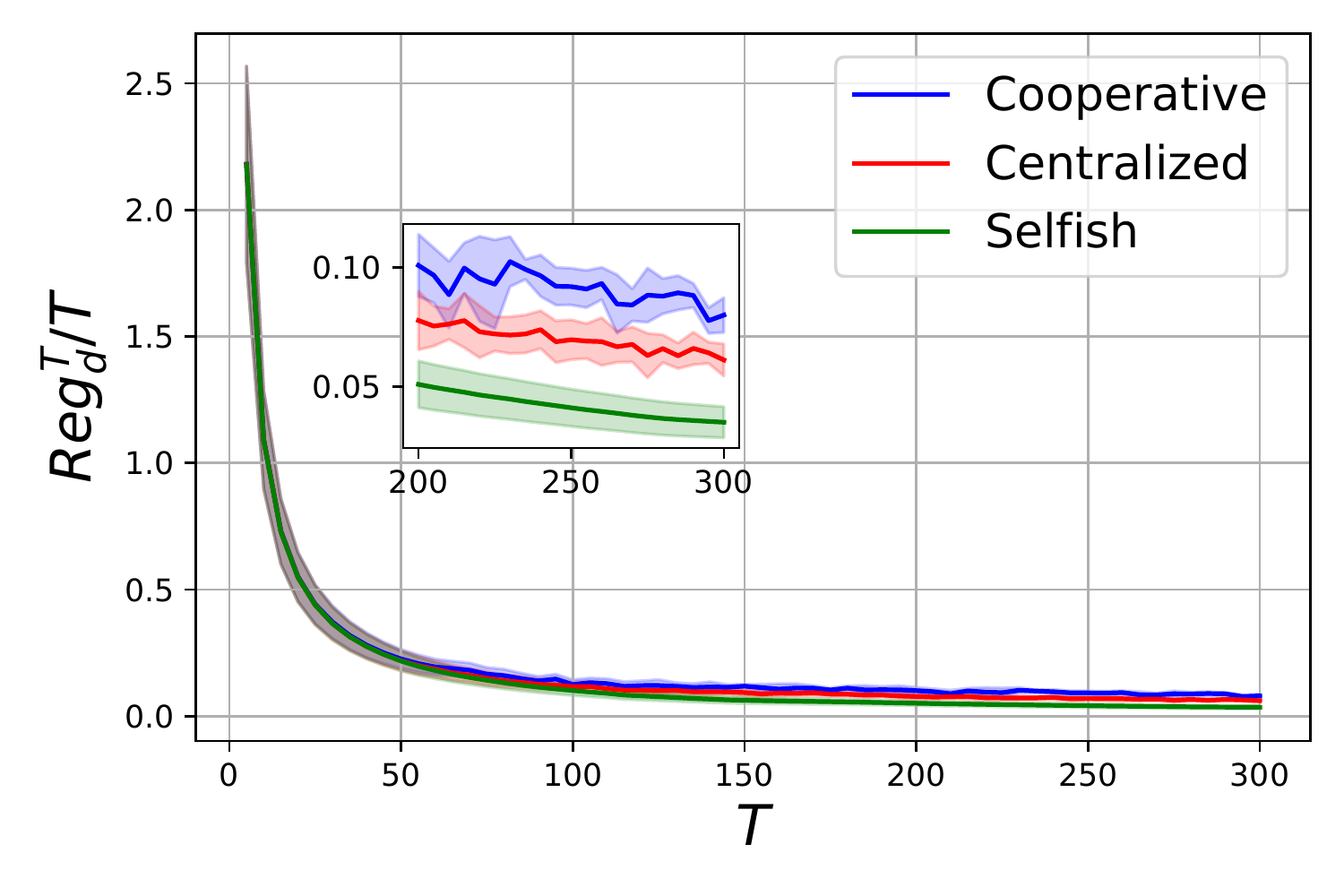}\label{fig:dynamic-milano}}
 \caption{Performance metrics vs horizon length $T$}
 \label{fig:results-oco}
 \end{figure}

Finally, we comment on the jump at $T\approx110$ in Fig.~\ref{fig:static-milano}, which is not present in Fig.~\ref{fig:dynamic-milano}. 
The difference of the two plots lies with the benchmarks; in Fig.~\ref{fig:gap}, the $y$-axis is in fact the cost gap between them, i.e. $\frac{1}{T}\sum_{t=1}^T \Big( f^t(x_*) - f^t(x^t_*) \Big)$; and in that plot we obviously see that same jump.
This behavior can be explained by the change of demand in Fig.~\ref{fig:demands}. Essentially, the static benchmark, for $T>110$, solves an optimization problem by finding a feasible $x_*$ for that extreme demand (at $T\approx110$), and will be then \emph{more constrained} compared to the dynamic benchmark, which solves the problem for each $t$ individually.

\begin{figure}
\centering
\subfigure[]{\includegraphics[width=0.48\columnwidth]{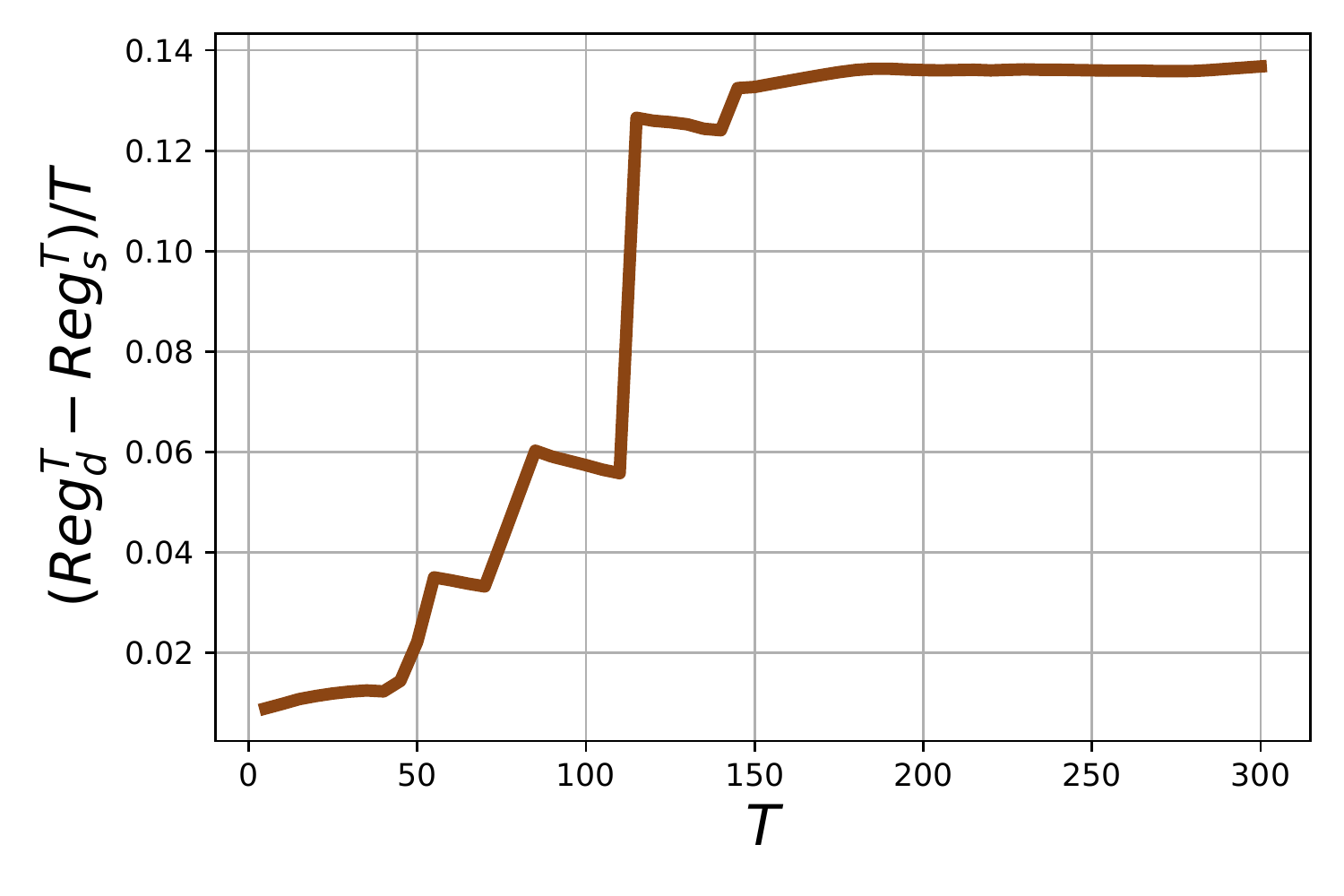}\label{fig:gap}}
\subfigure[]{\includegraphics[width=0.48\columnwidth]{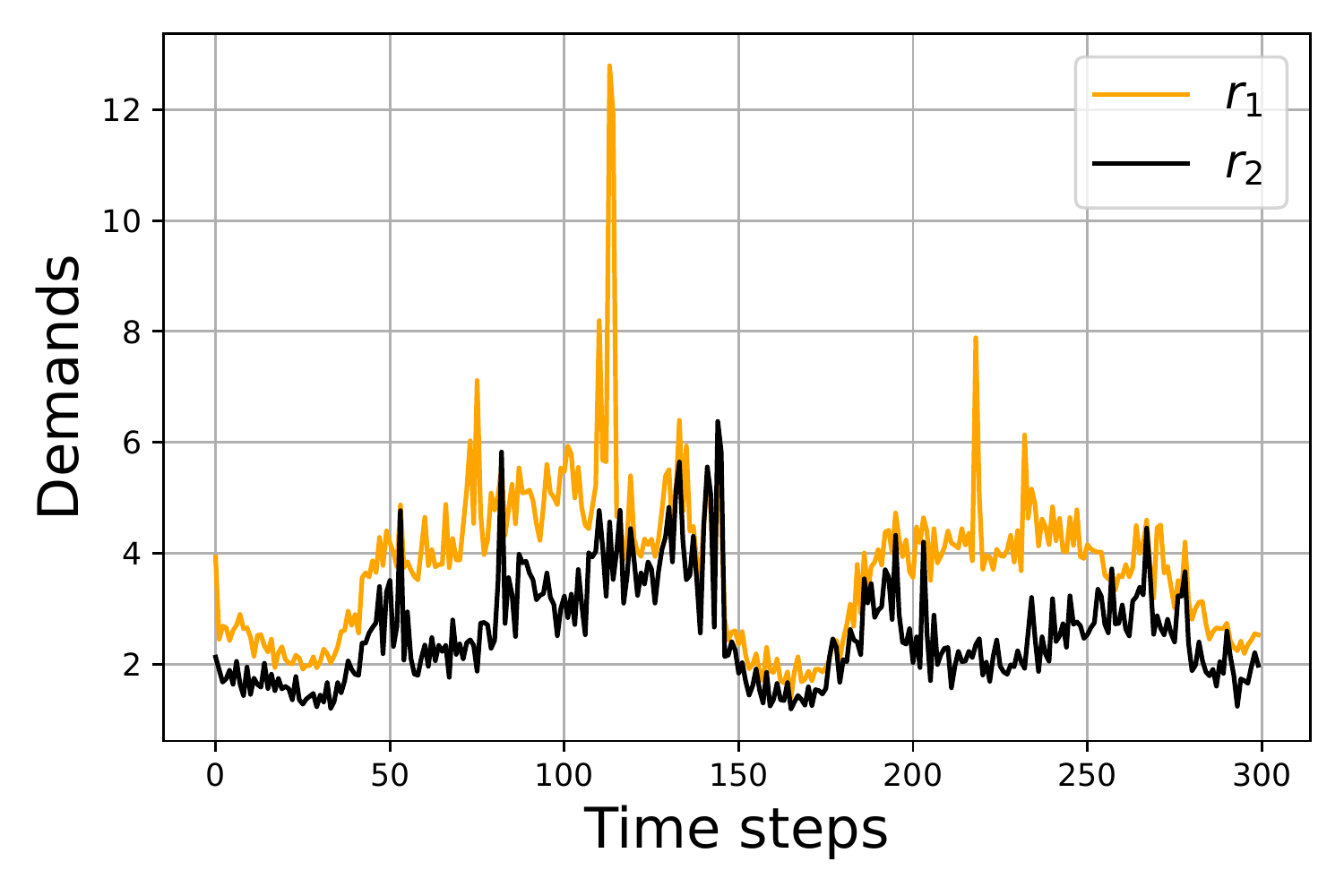}\label{fig:demands}}
\caption{Milano dataset: (a) Difference of regrets (and benchmarks) vs horizon $T$; (b) Demands of the devices over time.}
\label{fig:demands-benchmarks}
\end{figure}

	\section{Conclusions}
	\label{sec:conclusions}
	We revisit the problem of resource allocation in an IoT network, where devices can process part of the traffic requests, and/or offload the rest to more powerful computational entities (cloud, edge servers). 
The distributed nature of the setting, as well as the unpredictable environment, motivated us to model the network nodes (devices, BSs, servers) as distributed OCO actors. However, the nodes of the network are naturally coupled by flow conservation constraints at each node, which we model as long-term constraints, and as a result a fully decentralized algorithm is no longer possible.

In order to address this challenge, we propose a distributed OCO algorithm with limited communication between nodes, which practically leads to partially outdated gradient updates. Nevertheless, we show theoretically that our algorithm achieves sub-linear regret bound $\mathcal{O}(T^{1/2})$ and sub-linear constraint violation bound $\mathcal{O}(T^{3/4})$, which is the same order of bounds as a centralized algorithm for this setting. Numerical results
based on real data traces confirm our theoretical findings.
	
	\bibliographystyle{IEEEtran}
	\bibliography{ref}
	
	\newpage
	\appendices 
	\onecolumn
	
	\section{Proof of Theorem \ref{theorem}}\label{app:proof_regret_fit_normal_constraint_exchanging_information}

In this section, we provide the upper bounds of the static regret, the dynamic regret and the fit as a function of $T$.                                     
To that extent, we remind that:
\begin{align*}
	 f^t(\mathbf{x}^t)&=\sum_{d \in D}f_{d}^t(\mathbf{x}_{d}^t)+\sum_{b \in B}f_{b}^t(\mathbf{x}_{b}^t)+\sum_{s \in S}f_{s}^t(\mathbf{x}_{s}^t) \\
 	 \mathcal{L}^{t}(\mathbf{x}^{t},\boldsymbol{\lambda}^{t})&= f^t(\mathbf{x}^t)+ [\mathbf{h}^{t}(\mathbf{x}^{t})]^{\top}\boldsymbol{\lambda}^{t}, \quad \text{with} \;\; 
 	 \boldsymbol{\lambda}^t = \left((\boldsymbol{\lambda}_{d}^{t})_{d \in D},(\boldsymbol{\lambda}^{t}_{b})_{b \in B},(\boldsymbol{\lambda}_{s}^{t})_{s \in S}\right).
\end{align*}
For ease of understanding, we also recall that:
\begin{itemize}
	\item $\eta$ is the step-size of the decentralized algorithm.
	\item $|f^t_n|, \norm{\nabla f^t_n}\le F$; and $\norm{\nabla g^t_{n,i}}, \norm {\mathbf{g}^t_n}\le G$ for $n\in\mathcal{D,B,S}$.
	\item $R$ is the radius of the space set $\Omega$.
	\item $\sigma$ is a constant that we will define later. 
	\item $D,B,S$ is the number of devices, BSs and servers respectively.
	\item $N$ is the number of nodes, i.e. $N=D+B+S$.
	\item $M$ is the number of the constraint functions, with $M=D(B+1)+B(S+1)+S$.
	\item $E$ is the total number of edges, i.e. $E=DB+BS+S(S-1)$.
\end{itemize}

In order to prove the present theorem, let us first introduce the following Lemma.
\begin{lemma}\label{lem:lemma1_main_inequality}
	For any node $n$, the next inequality holds:
	\begin{align} \label{eq:proof_final_expression_device}
	(\mathbf{x}_{n}^{t}-\mathbf{x}_{n})^{\top}\nabla_{\mathbf{x}_{n}} \mathcal{L}^t(\mathbf{x}^t,\boldsymbol{\lambda}^t) &\le \frac{1}{2\eta}\left(\norm{\mathbf{x}_{n}-\mathbf{x}_{n}^{t}}^{2} - \norm{\mathbf{x}_{n}-\mathbf{x}_{n}^{t+1}}^{2}\right) \nonumber\\
	&\quad +2\eta F^2+2\eta(c_{1,n}+c_{2,n})  G^2 \norm{\boldsymbol{\lambda}^t}^2+ 2 \eta c_{2,n} G^2 \left(\norm{\boldsymbol{\lambda}^t}^2+\norm{\boldsymbol{\lambda}^{t-1}}^2 \right) \nonumber\\
	&\quad+\frac{3}{2}\eta \left(F^2+ c_{1,n} G^2\norm{\boldsymbol{\lambda}^{t-1}}^2+ c_{2,n} G^2 \norm{\boldsymbol{\lambda}^{t-2}}^2 \right)+\frac{\eta }{2}  c_{2,n} G^2 \norm{\boldsymbol{\lambda}^{t-1}}^2\nonumber\\
	&\quad+(\mathbf{x}_{n}^t-\mathbf{x}_{n})^{\top}\mathbf{H}^t_{nE}-(\mathbf{x}_{n}^{t-1}-\mathbf{x}_{n})^{\top} \mathbf{H}^{t-1}_{nE} \nonumber \\
	&= \frac{1}{2\eta}\left(\norm{\mathbf{x}_{n}-\mathbf{x}_{n}^{t}}^{2} - \norm{\mathbf{x}_{n}-\mathbf{x}_{n}^{t+1}}^{2}\right) + \frac{7}{2}\eta F^2 \nonumber\\
	&\quad + 2(c_{1,n}+2c_{2,n})\eta G^2 \norm{\boldsymbol{\lambda}^t}^2+ \frac{1}{2} (3c_{1,n}+5c_{2,n})\eta G^2 \norm{\boldsymbol{\lambda}^{t-1}}^2+ \frac{3}{2}c_{2,n} \eta  G^2 \norm{\boldsymbol{\lambda}^{t-2}}^2  \nonumber\\
	&\quad+(\mathbf{x}_{n}^t-\mathbf{x}_{n})^{\top}\mathbf{H}^t_{nE}-(\mathbf{x}_{n}^{t-1}-\mathbf{x}_{n})^{\top} \mathbf{H}^{t-1}_{nE}
	\end{align}
	\end{lemma}
	\noindent where $\mathbf{H}^{t}_{nE}$ is defined in \eqref{eq:gradient-device_external}-\eqref{eq:gradient-server_external}, $c_{1,n}$ is the number of local constraints at node $k$ and $c_{2,n}$ is the number of flow constraints from other nodes that include an optimization variable of node $k$. Specifically:
	\begin{equation}
		c_{1,n} = \begin{cases}
			B+1 &\text{for $n = d$}\\
			N+1 &\text{for $n = b$}\\
			1   &\text{for $n = s$}
		\end{cases}, \quad
		c_{2,n} = \begin{cases}
			B   &\text{for $n = d$}\\
			N   &\text{for $n = b$}\\
			N-1	&\text{for $n = s$}
		\end{cases}. \nonumber
	\end{equation}

\begin{proof}
	We prove in detail only the expression for the devices, as the respective expressions for BSs and servers is straight-forward following exactly the same steps. Thus, we have 
	
	\begin{align} 
		\norm{\mathbf{x}_{d}-\mathbf{x}_{d}^{t+1}}^{2}  &\overset{(a)}{\le} 
		\norm{\mathbf{x}_{d}-\Bigg( \mathbf{x}_{d}^t-\eta \bigg ( \nabla_{\mathbf{x}_{d}}f_{d}^t(\mathbf{x}_{d}^t)+\nabla_{\mathbf{x}_{d}}^{\top} [\mathbf{h}_{d}^t(\mathbf{x}_{d}^{t})]\boldsymbol{\lambda}^t_{d} +\sum_{b \in \mathcal{B}}\nabla_{\mathbf{x}_{d}}  h^{t-1}_{b0}(\mathbf{x}_{b}^{t-1};\mathbf{x}^{t-1}_\mathcal{D})\lambda^{t-1}_{b0} \bigg ) \Bigg ) }^2\nonumber \\
		&\overset{(b)}{=}\norm{\big (\mathbf{x}_{d}-\mathbf{x}_{d}^t \big )+\eta \Bigg [ \nabla_{\mathbf{x}_{d}} \mathcal{L}^t(\mathbf{x}^t,\boldsymbol{\lambda}^t)+\nabla_{\mathbf{x}_{d}}\bigg [\sum\limits_{b \in\mathcal{B}} \Big[h^{t-1}_{b0}(\mathbf{x}_{b}^{t-1};\mathbf{x}^{t-1}_{\mathcal{D}})\lambda_{b0}^{t-1}-h^t_{b0}(\mathbf{x}_{b}^{t};\mathbf{x}^{t}_\mathcal{D})\lambda_{b0}^{t} \Big ] \bigg ]\Bigg ]}^2 \nonumber \\		 			 		
		&\overset{(c)}{\le}\norm{\mathbf{x}_{d}-\mathbf{x}_{d}^{t}}^{2}+2\eta^2\norm{\nabla_{\mathbf{x}_{d}}\mathcal{L}^t(\mathbf{x}^t,\boldsymbol{\lambda}^t)}^2 +2\eta^2\norm{\nabla_{\mathbf{x}_{d}}\bigg [\sum\limits_{b \in\mathcal{B}} \Big [h^{t-1}_{b0}(\mathbf{x}_{b}^{t-1};\mathbf{x}^{t-1}_{\mathcal{D}})\lambda_{b0}^{t-1}-h^t_{b0}(\mathbf{x}_{b}^{t};\mathbf{x}^{t}_{\mathcal{D}})\lambda_{b0}^{t} \Big ]\bigg ]}^2 \nonumber\\
		&\quad +2\eta (\mathbf{x}_{d}-\mathbf{x}^{t}_{d})^{\top}\nabla_{\mathbf{x}_{d}} \mathcal{L}^t(\mathbf{x}^t,\boldsymbol{\lambda}^t) +2\eta (\mathbf{x}_{d}-\mathbf{x}_{d}^t)^{\top}\nabla_{\mathbf{x}_{d}} \bigg [\sum\limits_{b \in\mathcal{B}} \Big [h^{t-1}_{b0}(\mathbf{x}_{b}^{t-1};\mathbf{x}^{t-1}_{\mathcal{D}})\lambda_{b0}^{t-1}-h^t_{b0}(\mathbf{x}_{b}^{t};\mathbf{x}^{t}_{\mathcal{D}})\lambda_{b0}^{t}\Big ]\bigg ], 
	\end{align}
	where:
	\begin{itemize}
		\item (a) comes from the update rule of $\mathbf{x}_{d}^{t+1}$ and Assumption \ref{ass:oco} on the projection non-expansiveness.
		\item (b) uses the expression of  $\nabla_{\mathbf{x}_{d}}\mathcal{L}^t(\mathbf{x}^{t}, \boldsymbol{\lambda}^{t})$ from \eqref{eq:gradient-local_external}.
		\item (c) uses $\norm{\mathbf{x}+\mathbf{y}}^2=\norm{\mathbf{x}}^2+\norm{\mathbf{y}}^2+2\mathbf{x}^{\top}\mathbf{y}$ and then $\norm{\mathbf{x}+\mathbf{y}}^2\le  2(\norm{\mathbf{x}}^2+\norm{\mathbf{y}}^2)$.
	\end{itemize}
	Therefore, we have:
	\begin{align}\label{eq:basic_inequation_device}
		(\mathbf{x}_{d}^{t}-\mathbf{x}_{d})^{\top}\nabla_{\mathbf{x}_{d}} \mathcal{L}^t(\mathbf{x}^t,\boldsymbol{\lambda}^t) &\le \frac{1}{2\eta}\bigg(\norm{\mathbf{x}_{d}-\mathbf{x}_{d}^{t}}^{2} - \norm{\mathbf{x}_{d}-\mathbf{x}_{d}^{t+1}}^{2}\bigg)+\eta \norm{\nabla_{\mathbf{x}_{d}} \mathcal{L}^t(\mathbf{x}^t,\boldsymbol{\lambda}^t)}^2\nonumber\\
		&\quad +\eta\norm{\nabla_{\mathbf{x}_{d}}\sum\limits_{b \in\mathcal{B}}\big [ h^{t-1}_{b0}(\mathbf{x}_{b}^{t-1};\mathbf{x}^{t-1}_{\mathcal{D}})\lambda_{b0}^{t-1}-h^t_{b0}(\mathbf{x}_{b}^{t};\mathbf{x}^{t}_{\mathcal{D}})\lambda_{b0}^{t}\big ]}^2 \nonumber \\
		&\quad +(\mathbf{x}_{d}-\mathbf{x}_{d}^t)^{\top}\nabla_{\mathbf{x}_{d}} \bigg [\sum\limits_{b \in\mathcal{B}} \big [h^{t-1}_{b0}(\mathbf{x}_{b}^{t-1};\mathbf{x}^{t-1}_{\mathcal{D}})\lambda_{b}^{t-1}-h^t_{b0}(\mathbf{x}_{b}^{t};\mathbf{x}^{t}_{\mathcal{D}})\lambda_{b0}^{t}\big ]\bigg ].
	\end{align}
	We continue by analyzing each term of \eqref{eq:basic_inequation_device}. Specifically, we have:
	\begin{align}  \label{eq:gradient_lagran_funct_device}
		\eta \norm{\nabla_{\mathbf{x}_{d}} \mathcal{L}^t(\mathbf{x}^t,\boldsymbol{\lambda}^t)}^2&=\eta \norm{\nabla_{\mathbf{x}_{d}}f_{d}^t(\mathbf{x}_{d}^t)+\nabla_{\mathbf{x}_{d}}^{\top} [\mathbf{h}_{d}^t(\mathbf{x}_{d}^{t})]\boldsymbol{\lambda}^t_{d} +\sum_{b \in \mathcal{B}}\nabla_{\mathbf{x}_{d}}  h^t_{b0}(\mathbf{x}_{b}^{t};\mathbf{x}^{t}_{\mathcal{D}})\lambda^{t}_{b0}}^2\nonumber\\
		&\overset{(a)}{\le} \eta \Bigg (F+  \norm{\nabla_{\mathbf{x}_{d}}^{\top} [\mathbf{h}_{d}^t(\mathbf{x}_{d}^{t})]\boldsymbol{\lambda}^t_{d} +\sum_{b \in \mathcal{B}}\nabla_{\mathbf{x}_{d}}  h^t_{b0}(\mathbf{x}_{b}^{t};\mathbf{x}^{t}_{\mathcal{D}})\lambda^{t}_{b0}}\Bigg)^2  \nonumber\\ 
		&\overset{(b)}{\le} 2 \eta F^2+2\eta \bigg (\norm{\nabla_{\mathbf{x}_{d}}^{\top} [\mathbf{h}_{d}^t(\mathbf{x}_{d}^{t})]\boldsymbol{\lambda}^t_{d}} +\sum_{b \in \mathcal{B}}\norm{\nabla_{\mathbf{x}_{d}}  h^t_{b0}(\mathbf{x}_{b}^{t};\mathbf{x}^{t}_{\mathcal{D}})\lambda^{t}_{b0}} \bigg )^2 \nonumber \\
		&\overset{(c)}{\le} 2 \eta F^2+ 2\eta G^2\big ( |\lambda^t_{d0}| +\sum\limits_{b \in \mathcal{B}}  |\lambda_{db}^t|+\sum\limits_{b\in \mathcal{B}} |\lambda_{b0}^t|\big )^2 \nonumber \\
		&\overset{(d)}{\le} 2\eta F^2+2\eta (2B+1) G^2 \norm{\boldsymbol{\lambda}^t}^2		
	\end{align}
	where:
	\begin{itemize}
		\item (a) uses $\norm{\mathbf{x}+\mathbf{y}} \le \norm{\mathbf{x}}+\norm{\mathbf{y}}$ and then $\norm{\nabla{f^t_{d}(\mathbf{x}^t)}} \le F$.
		\item (b) uses first $\norm{\mathbf{x}+\mathbf{y}}^2 \le 2(\norm{\mathbf{x}}^2+\norm{\mathbf{y}}^2)$ and then $\norm{\mathbf{x}+\mathbf{y}} \le \norm{\mathbf{x}}+\norm{\mathbf{y}}$.
		\item (c) follows from $\norm{\nabla{h^t_{i}(\mathbf{x}^t)}} \le G \; \forall i \in [1,M]$.
		\item (d) uses $(\sum\limits_{i=1}^{K}a_i)^2 \le K\sum\limits_{i=1}^{K}a_i^2$ and the fact that $2B+1 \le M \Rightarrow |{\lambda}_{d0}^t|^2+\sum\limits_{b\in \mathcal{B}}|\lambda_{db}^t|^2 + \sum\limits_{b\in \mathcal{B}}|\lambda_{b0}^t|^2 \leq \norm{\boldsymbol{\lambda}^t}^2$.
	\end{itemize}  
	 
	Following the same procedure, i.e. steps (b)-(d), for the second term of \eqref{eq:basic_inequation_device}, we get:
	\begin{align}\label{eq:gradient_difference_constraint_device}
		\eta \norm{\nabla_{\mathbf{x}_{d}}\sum\limits_{b \in\mathcal{B}} \big [ h^{t-1}_{b0}(\mathbf{x}_{b}^{t-1};\mathbf{x}^{t-1}_{\mathcal{D}})\lambda_{b0}^{t-1}-h^t_{b0}(\mathbf{x}_{b}^{t};\mathbf{x}^{t}_{\mathcal{D}})\lambda_{b}^{t}\big ]}^2  
		&\le  2\eta G^2 \left( \Big( \sum\limits_{b \in\mathcal{B}} | \lambda_{b0}^{t-1}| \Big )^2+ \Big (\sum\limits_{b \in\mathcal{B}}|\lambda_{b0}^{t}|\Big )^2\right)\nonumber\\
		&\le 2\eta B G^2 \Big (\norm{\boldsymbol{\lambda}^t}^2+\norm{\boldsymbol{\lambda}^{t-1}}^2 \Big )
	\end{align}
	
	For the last term of \eqref{eq:basic_inequation_device}, if we add and subtract the term $(\mathbf{x}_{d}^{t-1})^{\top}\nabla_{\mathbf{x}_{d}} \bigg [\sum\limits_{b \in\mathcal{B}} h^{t-1}_{b0}(\mathbf{x}_{b}^{t-1};\mathbf{x}^{t-1}_{\mathcal{D}})\lambda_{b0}^{t-1} \bigg ]$ we get:
	\begin{align}\label{eq:part_telescopic}
		& (\mathbf{x}_{d}-\mathbf{x}_{d}^t)^{\top}\nabla_{\mathbf{x}_{d}} \bigg [\sum\limits_{b \in\mathcal{B}} \big [h^{t-1}_{b0}(\mathbf{x}_{b}^{t-1};\mathbf{x}^{t-1}_{\mathcal{D}})\lambda_{b0}^{t-1}-h^t_{b0}(\mathbf{x}_{b}^{t};\mathbf{x}^{t}_{\mathcal{D}})\lambda_{b0}^{t}\big ]\bigg ] = (\mathbf{x}_{d}-\mathbf{x}_{d}^{t-1})^{\top}\nabla_{\mathbf{x}_{d}} \bigg [\sum\limits_{b \in\mathcal{B}} h^{t-1}_{b0}(\mathbf{x}_{b}^{t-1};\mathbf{x}^{t-1}_{\mathcal{D}})\lambda_{b0}^{t-1} \bigg ] \nonumber \\		
		&\quad -(\mathbf{x}_{d}-\mathbf{x}_{d}^{t})^{\top}\nabla_{\mathbf{x}_{d}} \bigg [\sum\limits_{b \in\mathcal{B}}h^t_{b0}(\mathbf{x}_{b}^{t};\mathbf{x}^{t}_{\mathcal{D}})\lambda_{b0}^{t}\bigg ] +(\mathbf{x}_{d}^{t-1}-\mathbf{x}_{d}^{t})^{\top}\nabla_{\mathbf{x}_{d}} \bigg [\sum\limits_{b \in\mathcal{B}} h^{t-1}_{b0}(\mathbf{x}_{b}^{t-1};\mathbf{x}^{t-1}_{\mathcal{D}})\lambda_{b0}^{t-1} \bigg ].
	\end{align}
	For the last term of \eqref{eq:part_telescopic} we have:
	\begin{align} \label{eq:inner_product_device}
		&(\mathbf{x}_{d}^{t-1}-\mathbf{x}_{d}^{t})^{\top}\nabla_{\mathbf{x}_{d}} \bigg [\sum\limits_{b \in\mathcal{B}} h^{t-1}_{b0}(\mathbf{x}_{b}^{t-1};\mathbf{x}^{t-1}_{\mathcal{D}})\lambda_{b0}^{t-1} \bigg ] \overset{(a)}{\le}
		\frac{1}{2\eta} \norm{\mathbf{x}_{d}^t-\mathbf{x}_{d}^{t-1}}^2+\frac{\eta}{2} \norm{\nabla_{\mathbf{x}_{d}} \bigg [\sum\limits_{b \in\mathcal{B}} h^{t-1}_{b0}(\mathbf{x}_{b}^{t-1};\mathbf{x}^{t-1}_{\mathcal{D}})\lambda_{b0}^{t-1} \bigg ]}^2 \nonumber \\
		&\quad \overset{(b)}{\le} \frac{\eta}{2} \norm{\nabla_{\mathbf{x}_{d}}f_d^{t-1}(\mathbf{x}_{d}^{t-1})+\nabla^{\top}_{\mathbf{x}_{d}} [\mathbf{h}^{t-1}_{d}(\mathbf{x}_{d}^{t-1})]\boldsymbol{\lambda}_{d}^{t-1} +\sum\limits_{b \in \mathcal{B}} \nabla_{\mathbf{x}_{d}}[h^{t-2}_{b0}(\mathbf{x}_{b}^{t-2}; \mathbf{x}^{t-2}_{\mathcal{D}})] \lambda_{b0}^{t-2}}^2+\frac{\eta}{2}  B G^2 \norm{\boldsymbol{\lambda}^{t-1}}^2 \nonumber \\
		&\quad \overset{(c)}{\le} \frac{3\eta}{2} \left( F^2+(B+1)G^2 \norm{\boldsymbol{\lambda}^{t-1}}^2+ BG^2 \norm{\boldsymbol{\lambda}^{t-2}}^2\right) +\frac{\eta}{2} B G^2 \norm{\boldsymbol{\lambda}^{t-1}}^2,
	\end{align}
		where:
	\begin{itemize}
		\item (a) uses $ \mathbf{x}^{\top}\mathbf{y} \leq \frac{1}{2\eta}\left(\norm{\mathbf{x}}^2 + \eta^2\norm{\mathbf{y}}^2\right)$ 
		\item (b) comes from the update rule of $\mathbf{x}_{d}^{t}$ and Assumption \ref{ass:oco}.  
		\item (c) uses  $\norm{\mathbf{x}+\mathbf{y}+\mathbf{z}}^2 \le 3 (\norm{\mathbf{x}}^2+\norm{\mathbf{y}}^2+\norm{\mathbf{z}}^2)$
	\end{itemize} 

	Combining \eqref{eq:basic_inequation_device}-\eqref{eq:inner_product_device} we directly obtain \eqref{eq:proof_final_expression_device}. Following the same steps for BSs and servers completes the proof of Lemma \ref{lem:lemma1_main_inequality}.
\end{proof}

The next step to prove the theorem is to combine the results of Lemma \ref{lem:lemma1_main_inequality} for all nodes.
According to Assumption \ref{ass:oco}, $\mathcal{L}^t(.,\boldsymbol{\lambda})$ is a convex function with $\mathbf{x}$ and thus, we have
\begin{align}
	\mathcal{L}^t(\mathbf{x}^t,\boldsymbol{\lambda}^t)-\mathcal{L}^t(\mathbf{x},\boldsymbol{\lambda}^t)& \le (\mathbf{x}^t - \mathbf{x})^{\top} \nabla_x \mathcal{L}^t(\mathbf{x}^t,\boldsymbol{\lambda}^t) \nonumber \\
	&=\sum\limits_{d \in \mathcal{D}} (\mathbf{x}_{d}^t - \mathbf{x}_{d})^{\top} \nabla_{\mathbf{x}_{d}} \mathcal{L}^t(\mathbf{x}^t,\boldsymbol{\lambda}^t)+ \sum\limits_{b \in \mathcal{B}}(\mathbf{x}_{b}^t-\mathbf{x}_{b})^{\top} \nabla_{\mathbf{x}_{b}} \mathcal{L}^t(\mathbf{x}^t,\boldsymbol{\lambda}^t)+ \sum\limits_{s \in \mathcal{S}} (\mathbf{x}_{s}^t - \mathbf{x}_{s})^{\top} \nabla_{\mathbf{x}_{s}} \mathcal{L}^t(\mathbf{x}^t,\boldsymbol{\lambda}^t).
\end{align}

Then, from Lemma \ref{lem:lemma1_main_inequality} we get:
\begin{align} \label{eq:summing-the-lemma-terms}
	\mathcal{L}^t(\mathbf{x}^t, \boldsymbol{\lambda}^t) - \mathcal{L}^t(\mathbf{x}, \boldsymbol{\lambda}^t) &\le \frac{1}{2\eta}\Big(\norm{\mathbf{x}-\mathbf{x}^{t}}^{2} - \norm{\mathbf{x}-\mathbf{x}^{t+1}}^{2}\Big) + \frac{7}{2} \eta N F^2 \nonumber \\
	&\quad +  \Big( (6DB+2D) + (6BS+2B) + (4S^2-2S) \Big)\eta G^2 \norm{\boldsymbol{\lambda}^t}^2 \nonumber \\
	&\quad +  \Big (  (4DB+\frac{3}{2}D) + (4BS+ \frac{3}{2}B) + (\frac{5}{2}S^2-S) \Big) \eta G^2 \norm{\boldsymbol{\lambda}^{t-1}}^2 \nonumber\\
	&\quad +  \frac{3}{2}\Big(  DB+ BS+ S(S-1) \Big) \eta G^2 \norm{\boldsymbol{\lambda}^{t-2}}^2 + Q^{t}(\mathbf{x})-Q^{t-1}(\mathbf{x})
\end{align}
where
\begin{align} \label{eq:qfunc}
	Q^t(\mathbf{x}) &   :=   \sum\limits_{d \in \mathcal{D}}\Bigg [(\mathbf{x}_{d}^t-\mathbf{x}_{d})^{\top}\nabla_{\mathbf{x}_{d}} \bigg [\sum\limits_{b \in\mathcal{B}} h^t_{b0}(\mathbf{x}_{b}^{t};\mathbf{x}^{t}_{\mathcal{D}})\lambda_{b0}^{t} \bigg ] \Bigg ] + \sum\limits_{b \in \mathcal{B}} \Bigg [(\mathbf{x}_{b}^{t}-\mathbf{x}_{b})^{\top}\nabla_{\mathbf{x}_{b}} \bigg [\sum\limits_{s \in \mathcal{S}} h^t_{s0}(\mathbf{x}^{t}_{s};\mathbf{x}^{t}_{\mathcal{B}},\mathbf{x}^{t}_{\mathcal{S}_{-s}})\lambda_{s0}^{t}\bigg ] \Bigg ]\nonumber \\
	&\quad+\sum\limits_{s \in \mathcal{S}} \Bigg[ (\mathbf{x}_{s}^{t}-\mathbf{x}_{s})^{\top}\nabla_{\mathbf{x}_{s}} \bigg [\sum\limits_{s' \in \mathcal{S}_{-s}} h_{s'0}^{t}(\mathbf{x}_{s'}^{t};\mathbf{x}_{\mathcal{B}}^{t},\mathbf{x}^{t}_{\mathcal{S}_{-s'}})\lambda_{s'0}^{t}\bigg ] \Bigg ] 
\end{align}

To upper-bound the multiplicative terms in front of the multipliers in (\ref{eq:summing-the-lemma-terms}), we define
\begin{align}
K=2D(3B+1) + 2B(3S+1) + 2S(2S-1), 
\end{align}
and get
\begin{align}\label{eq:lagranigan_function_inequality}
	\mathcal{L}^t(\mathbf{x}^t,\boldsymbol{\lambda}^t)-\mathcal{L}^t(\mathbf{x},\boldsymbol{\lambda}^t) 
	& \le 
	\frac{1}{2\eta}  \underbrace{ \Big(\norm{\mathbf{x} - \mathbf{x}^{t}}^{2} - \norm{\mathbf{x}-\mathbf{x}^{t+1}}^{2}\Big)  }_{\text{(a)}} + \underbrace{ Q^{t}(\mathbf{x})-Q^{t-1}(\mathbf{x}) }_{\text{(b)}} + \frac{7}{2} \eta N F^2 \nonumber\\
	&\quad + \eta K G^2  \underbrace{ \Big (\norm{\boldsymbol{\lambda}^t}^2+\norm{\boldsymbol{\lambda}^{t-1}}^2+\norm{\boldsymbol{\lambda}^{t-2}}^2\Big )	}_{\text{(c)}}.
\end{align}

We then sum from $t=1 \to T$ and upper-bound the (a)-(c) terms of \eqref{eq:lagranigan_function_inequality} using the following:
\begin{itemize}
\item (a) $\sum\limits_{t=1}^{T}( \norm{\mathbf{x} - \mathbf{x}^{t}}^{2} - \norm{\mathbf{x} - \mathbf{x}^{t+1}}^{2}) = \norm{\mathbf{x} - \mathbf{x}^1}^2 -\norm{\mathbf{x} - \mathbf{x}^{T+1}}^2 \le \norm{\mathbf{x}^1 - \mathbf{x}}^2 \le R^2$;
\item (b) $\sum\limits_{t=1}^{T}(Q^{t}(\mathbf{x}) - Q^{t-1}(\mathbf{x})) = Q^{T}(\mathbf{x}) - Q^{0}(\mathbf{x}) \le \norm{Q^{T}(\mathbf{x})} + \norm{Q^{0}(\mathbf{x})} \le \frac{2REG^2}{\eta \sigma} $, where in the last step we use $\norm{Q^{t}(\mathbf{x})} \le \left(DB+BS+S(S-1)\right)RG \norm{\boldsymbol{\lambda}^t}$ according to \eqref{eq:qfunc} and the update rule $\boldsymbol{\lambda}^t = \frac{ \mathbf{h}^{t}({\mathbf{x}}^t) }{ \eta \sigma }$;
\item (c) $ \sum\limits_{t=1}^{T} (\norm{\boldsymbol{\lambda}^t}^2 + \norm{\boldsymbol{\lambda}^{t-1}}^2 + \norm{\boldsymbol{\lambda}^{t-2}}^2) =
3 \sum\limits_{t=1}^{T} \norm{\boldsymbol{\lambda}^t}^2 - 2  \norm{\boldsymbol{\lambda}^T}^2 - \norm{\boldsymbol{\lambda}^{T-1}}^2$, assuming $\norm{\boldsymbol{\lambda}^0}=0$ and $\norm{\boldsymbol{\lambda}^{-1}}=0$.
\end{itemize}

Putting everything together we have:
\begin{equation}\label{eq:bound_lagrangian_function}
	\sum\limits_{t=1}^{T} \left(\mathcal{L}^t(\mathbf{x}^t, \boldsymbol{\lambda}^t) - \mathcal{L}^t(\mathbf{x}, \boldsymbol{\lambda}^t)\right) 
	\le
	\frac{R^2}{2 \eta} +  \frac{2REG^2}{\eta \sigma}  + 
	 \frac{7}{2}\eta N F^2 T + 3\eta K G^2  \sum\limits_{t=1}^{T} \norm{\boldsymbol{\lambda}^t}^2 \\
\end{equation}

\subsection{Static Regret}
If we set $x=x_{*}$ for which $\mathbf{h}^t(\mathbf{x}_*)=0$ and use $\boldsymbol{\lambda}^t = \frac{ \mathbf{h}^{t}({\mathbf{x}}^t) }{ \eta \sigma }$, we get
\begin{align}
	\sum\limits_{t=1}^{T}\Big(f^t(\mathbf{x}^t) + \frac{\norm{\mathbf{h}^t(\mathbf{x}^t)}^2}{\eta \sigma} - f^t(\mathbf{x}_*) \Big) 
	 \le 
	\frac{R^2}{2 \eta} + \frac{2REG^2}{\eta \sigma} + \frac{7}{2}\eta N F^2 T + 3K G^2  \sum\limits_{t=1}^{T} \frac{\norm{\mathbf{h}^t(\mathbf{x}^t)}^2}{\eta\sigma^2}, 
\end{align}

that yields
\begin{align}\label{eq:bound_lagrag_function_after_elimination_of_lambda_sharing_info}
	\sum\limits_{t=1}^{T}\Big( f^t(\mathbf{x}^t)-f^t(\mathbf{x}_*) \Big) + \frac{1}{\eta \sigma}\sum\limits_{t=1}^{T}\norm{\mathbf{h}^{t}({\mathbf{x}^t})}^2 (1-\frac{3KG^2}{\sigma})
	 \le \frac{R^2}{2 \eta} + \frac{2REG^2}{\eta \sigma} + \frac{7}{2}\eta N F^2 T.
\end{align}

For $\sigma > 3KG^2$, the upper bound $U_{sr}$ of the static regret  is found by 
\begin{equation}\label{eq:static_regret_bound}
	\sum\limits_{t=1}^{T}\Big(f^t(\mathbf{x}^t) - f^t(\mathbf{x}_*)\Big) \le \frac{R^2}{2 \eta} + \frac{2REG^2}{\eta \sigma}  + \frac{7}{2}\eta N F^2 T \triangleq U_{sr}.
\end{equation}

\subsection{Dynamic Regret}
Choosing again $\sigma > 3KG^2$ and plugging the instantaneous optimal solution $x=x^t_*$ in \eqref{eq:lagranigan_function_inequality}, we follow the procedure we used for the static regret and obtain:
\begin{equation}
	\sum\limits_{t=1}^{T}\Big(f^t(\mathbf{x}^t)-f^t(\mathbf{x}^t_*) \Big)  \le \frac{1}{2\eta}\sum\limits_{t=1}^{T}(  \norm{\mathbf{x}_*^t-\mathbf{x}^{t}}^{2} - \norm{\mathbf{x}_*^t-\mathbf{x}^{t+1}}^{2}) + \frac{2REG^2}{\eta \sigma} + \frac{7}{2}\eta N F^2 T.
\end{equation}

The only term that needs different handling is $\norm{\mathbf{x}_*^t - \mathbf{x}^{t}}^{2} - \norm{\mathbf{x}_*^t - \mathbf{x}^{t+1}}^{2}$, for which we can write
\begin{align}
	\norm{\mathbf{x}_*^t - \mathbf{x}^{t}}^{2} - \norm{\mathbf{x}_*^t - \mathbf{x}^{t+1}}^{2}  =   \underbrace{\norm{\mathbf{x}_*^t - \mathbf{x}^{t}}^{2} - \norm{\mathbf{x}_*^{t-1} -  \mathbf{x}^{t}}^{2}}_{\text{(a)}}
	 + 
	 \underbrace{\norm{\mathbf{x}_*^{t-1} - \mathbf{x}^{t}}^{2} - \norm{\mathbf{x}_*^t-\mathbf{x}^{t+1}}^{2}}_{\text{(b)}} .
\end{align}

We are interested in the telescopic sums of (a, b). In particular, for (a) we use $\mathbf{x}^2-\mathbf{y}^2=(\mathbf{x}-\mathbf{y})^{\top}(\mathbf{x}+\mathbf{y})$ and get
\begin{align}
\sum\limits_{t=1}^{T} \left(\norm{\mathbf{x}_*^t - \mathbf{x}^{t}}^{2} - \norm{\mathbf{x}_*^{t-1}-\mathbf{x}^{t}}^{2}\right) & = \sum\limits_{t=1}^{T} (\mathbf{x}_*^t - \mathbf{x}_*^{t-1})^{\top} (\mathbf{x}_*^{t} + \mathbf{x}_*^{t-1} - 2 \mathbf{x}^t) 
\le 
\sum\limits_{t=1}^{T} \norm{\mathbf{x}_*^{t} + \mathbf{x}_*^{t} - 2 \mathbf{x}^t} \norm{\mathbf{x}_*^t - \mathbf{x}_*^{t-1}} \nonumber\\
& \le
2R \sum\limits_{t=1}^{T} \norm{\mathbf{x}_*^t - \mathbf{x}_*^{t-1}} = 2RV(\mathbf{x}_*^{1:T}) 
\end{align}
where $V(\mathbf{x}_*^{1:T})$ is the sum of distances of consecutive optimal solutions. Moreover, for term (b) we have 
\begin{equation}
	 \sum\limits_{t=1}^{T}( \norm{\mathbf{x}_*^{t-1}-\mathbf{x}^{t}}^{2}- \norm{\mathbf{x}_*^t-\mathbf{x}^{t+1}}^{2})=\norm{\mathbf{x}_*^{0} - \mathbf{x}^{1}}^{2}- \norm{\mathbf{x}_*^T-\mathbf{x}^{T+1}}^{2} \le R^2.
\end{equation}

Finally, the upper bound $U_{dr}$ of the dynamic regret is 
\begin{equation}\label{eq:dynamic_regret_bound}
	\sum\limits_{t=1}^{T}\Big(f^t(\mathbf{x}^t) - f^t(\mathbf{x}^t_*) \Big) \le \frac{R}{\eta}   V(\mathbf{x}_*^{1:T}) + \frac{R^2}{2\eta} + \frac{2REG^2}{\eta \sigma} + \frac{7}{2}\eta N F^2 T \triangleq U_{dr}.
\end{equation}
By combining \eqref{eq:static_regret_bound},\eqref{eq:dynamic_regret_bound}, the relation to the bound of the static regret is $U_{dr} = U_{sr} + \frac{R}{\eta} V(\mathbf{x}_*^{1:T})$.

\subsection{Fit}
In order to bound the fit, we choose $\sigma > 3KG^2$ in \eqref{eq:bound_lagrag_function_after_elimination_of_lambda_sharing_info} and use $\sum\limits_{t=1}^{T}\left(f^t(\mathbf{x}_*)  -  f^t(\mathbf{x}^t)\right) \le 2NFT$, which gives 
\begin{equation}
	\sum\limits_{t=1}^{T}\norm{\mathbf{h}^{t}({\mathbf{x}^t})}^2  \le  \frac{1}{\beta} ( \frac{\sigma R^2}{2} + 2REG^2 + \frac{7}{2}\eta^2 \sigma N F^2 T + 2 \eta \sigma N F T), 
\end{equation}
where $\beta=1-\frac{3K G^2}{\sigma}$. By further using the property $\Big(\sum\limits_{i=1}^{n} a_i \Big)^2 \le n \sum\limits_{i=1}^{n} a_i^2$, we obtain
\begin{align}\label{eq:fit-bound}
	\Big( \sum\limits_{t=1}^{T}\norm{ \mathbf{h}^{t}(\mathbf{x}^t) } \Big)^2  \le  T  \sum\limits_{t=1}^{T}\norm{ \mathbf{h}^{t}(\mathbf{x}^t) }^2  
	& \le 
	 \frac{T}{\beta}  ( \frac{\sigma R^2}{2} + 2REG^2 + \frac{7}{2}\eta^2 \sigma N F^2 T + 2 \eta \sigma N F T) .
\end{align}

Taking the square root on both sides of \eqref{eq:fit-bound} 
\begin{align} 
	  \sum\limits_{t=1}^{T}\norm{ \mathbf{h}^{t}(\mathbf{x}^t) }  \le 
	 \Big( \frac{T}{\beta}  (  \frac{\sigma R^2}{2} + 2REG^2 + \frac{7}{2}\eta^2 \sigma N F^2 T + 2 \eta \sigma N F T ) \Big)^{1/2} 
\end{align}

Finally, $\sum_{m=1}^M h_m^t(\mathbf{x}^t) \le \sqrt{  M  \sum_{m=1}^M (h_m^t(\mathbf{x}^t))^2  } = \sqrt{M} \norm{\mathbf{h}^t(\mathbf{x}^t)}$  which gives us the upper bound $U_f$ on the fit
\begin{align}\label{eq:bound_fit_final}
	\sum\limits_{t=1}^{T}\sum\limits_{m=1}^{M}h_m^{t}(\mathbf{x}^t) \le \Big( \frac{MT}{\beta}  (  \frac{\sigma R^2}{2} + 2REG^2 + \frac{7}{2}\eta^2 \sigma N F^2 T + 2 \eta \sigma N F T ) \Big)^{1/2}.
\end{align}
By combining \eqref{eq:static_regret_bound},\eqref{eq:bound_fit_final}, the relation to the bound of the static regret is $U_{f} =  \sqrt{\frac{\eta \sigma }{\beta}MT(U_{sr} + 2 N F T )}$.


\section{Proof for implications of Assumption~\ref{ass:oco}} \label{app:implications}

Here we show formally two of the implications of Assumption~\ref{ass:oco}, namely (A) that $f(\mathbf{x})$ is bounded and (B) that clipped function $h(\mathbf{x})$ is convex. Below, for simplicity we drop the node $n$ and timestep $t$ sub and superscripts.

\subsection{Function $f(\mathbf{x})$ is bounded}

We want to show that for a finite $f$ for which Assumption~\ref{ass:oco} holds, there exists $F$ such that $|f| \le F$ for all $\mathbf{x}\in\Omega$.
From $\norm{\nabla_{\boldsymbol{x}} f} \le F^{\prime}$, it is implied that for any $\boldsymbol{x},\boldsymbol{y} \in \Omega$, we have
\begin{equation}
|f(\boldsymbol{x}) - f(\boldsymbol{y})| \le F^{\prime} \underbrace{\norm{\boldsymbol{x} - \boldsymbol{y}}}_{\le R} \le F^{\prime} R. \nonumber
\end{equation}

Furthermore, using that $|f(\boldsymbol{x})| - | f(\boldsymbol{y})| \le |f(\boldsymbol{x}) - f(\boldsymbol{y})| \le F^{\prime}R$, we get
\begin{equation}
|f(\boldsymbol{x})| \le F^{\prime} R + |f(\boldsymbol{y})| = F, \nonumber
\end{equation}
where we do not know $|f(\boldsymbol{y})|$, but we know that by definition it is finite.
Since Assumption~\ref{ass:oco} holds for $g$, the exact same steps can be followed in order to show that $g$ is bounded.

\subsection{Convexity of $h(\mathbf{x})$}

Recall the definitions of $h(\mathbf{x})$ and its gradient

\begin{align}
	h(\mathbf{x})=[g(\mathbf{x})]^+=
	\begin{cases}
		0 & \text{if $g(\mathbf{x}) \le 0$} \\
		g(\mathbf{x}) & \text{if $g(\mathbf{x}) > 0$}
	\end{cases}
	~~~\text{and}~~~
		\nabla h(\mathbf{x})=\nabla[g(\mathbf{x})]^+=
	\begin{cases}
		\mathbf{0} & \text{if $g(\mathbf{x}) \le 0$} \\
		\nabla g(\mathbf{x}) & \text{if $g(\mathbf{x}) > 0$}
	\end{cases} \nonumber
\end{align}

We want to show that if $g(\mathbf{x})$ is convex, that is, $g(\mathbf{x}) - g(\mathbf{y}) \le (\mathbf{x}-\mathbf{y})^{\top}\nabla g(\mathbf{x})$, then for any $\mathbf{x},\mathbf{y} \in\Omega$, we have 
\begin{equation}
h(\mathbf{x}) - h(\mathbf{y}) \le (\mathbf{x}-\mathbf{y})^{\top}\nabla h(\mathbf{x}).\nonumber
\end{equation}

In what follows, we only use the above definitions of $h$ and its gradient, and the fact that $g$ is convex. The four possible cases are the following.

\myitem{(a):} If $g(\mathbf{x}) > 0$ and $g(\mathbf{y}) > 0$, then
$
h(\mathbf{x}) - h(\mathbf{y}) = g(\mathbf{x})-g(\mathbf{y}) \le (\mathbf{x}-\mathbf{y})^{\top}\nabla g(\mathbf{x})  = (\mathbf{x}-\mathbf{y})^{\top}\nabla h(\mathbf{x}).
$

\myitem{(b):} If $g(\mathbf{x}) \le 0$ and $g(\mathbf{y}) \le 0$, then
$
h(\mathbf{x}) - h(\mathbf{y}) = 0 - 0 \le 0 = (\mathbf{x} - \mathbf{y})^{\top} \nabla h(\mathbf{x})$.

\myitem{(c):} If $g(\mathbf{x}) > 0$ and $g(\mathbf{y}) \le 0$, then
$
h(\mathbf{x}) - h(\mathbf{y}) = h(\mathbf{x}) = g(\mathbf{x}) \le g(\mathbf{x}) - g(\mathbf{y}) \le (\mathbf{x}-\mathbf{y})^{\top}\nabla g(\mathbf{x}) = (\mathbf{x} - \mathbf{y})^{\top}\nabla h(\mathbf{x}).
$

\myitem{(d):} If $g(\mathbf{x}) \le 0$ and $g(\mathbf{y}) > 0$, then
$
h(\mathbf{x})-h(\mathbf{y}) = -h(\mathbf{y}) = -g(\mathbf{y})\le 0 = (\mathbf{x} - \mathbf{y})^{\top} \nabla h(\mathbf{x}).
$

\end{document}